\documentclass[11pt,letterpaper]{article}

\usepackage[margin=1in]{geometry}

\usepackage{tweaklist}
\usepackage{amsmath,amsfonts,graphpap,amscd,mathrsfs,graphicx,lscape}
\usepackage{epsfig,amssymb,amstext,xspace}
\usepackage{float}	
\usepackage{hyperref}

\usepackage{color}              
\usepackage{epsfig}

\usepackage[suppress]{color-edits}
\addauthor{cd}{green}
\addauthor{vs}{red}

\usepackage{amsthm}

\newcommand{\E}{\mathbb{E}}

\newcommand{\rev}{\mathcal{R}}

\newtheorem{theorem}{Theorem}
\newtheorem{lemma}[theorem]{Lemma}

\newtheorem{corollary}[theorem]{Corollary}

\newtheorem{defn}{Definition}
\newtheorem{algorithm}{Algorithm}

\newenvironment{rtheorem}[3][]{
\bigskip
\noindent \ifthenelse{\equal{#1}{}}{\bf #2 #3}{\bf #2 #3 (#1)}
\begin{it}
}{\end{it}}

\def\squareforqed{\hbox{\rule{2.5mm}{2.5mm}}}

\def\QED{\ifmmode\squareforqed 
  \else{\nobreak\hfil   
    \penalty50                 
    \hskip1em                  
    \null                      
    \nobreak                   
    \hfil                      
    \squareforqed              
    \parfillskip=0pt           
    \finalhyphendemerits=0     
    \endgraf}                  
  \fi}

\def\blksquare{\rule{2mm}{2mm}}
\def\qedsymbol{\blksquare}
\newcommand{\bg}[1]{\medskip\noindent{\bf #1}}
\newcommand{\ed}{{\hfill\qedsymbol}\medskip}
\newenvironment{proofsk}{\noindent\textbf{Proof sketch.}}{\QED}
\newenvironment{proofskof}[1]{\noindent\textbf{Proof sketch of #1.}}{\QED}
\newenvironment{proofof}[1]{{\vspace{2em}\noindent\textbf{Proof of #1. }}}{\ed}

\newcommand{\R}{\ensuremath{\mathbb R}}

\newcommand{\N}{\ensuremath{\mathbb N}}
\newcommand{\A}{\ensuremath{\mathcal{A}}}

\newcommand{\F}{\ensuremath{\mathcal F}}

\newcommand{\T}{\ensuremath{\mathcal T}}

\newcommand{\Ell}{\ensuremath{\mathcal{L}}}

\newcommand{\poly}{\operatorname{poly}}

\newcommand{\comment}[1]{}
 {}

\newcommand{\junk}[1]{}

\newlength{\tmp} \newlength{\lpsx} \newlength{\lpsy} \newlength{\upsx} \newlength{\upsy}

\newcommand{\opt}{\text{\textsc{Opt}} }

\newcommand{\Or}[1]{\ensuremath{M\left(#1\right)}}

\newcommand{\Omit}[1]{}

\newcommand{\RP}{\ensuremath{{\tt RP}}}

\newcommand{\NP}{\ensuremath{{\tt NP}}}

\newcommand{\FTPL}{\ensuremath{{\tt FTPL}}}
\newcommand{\BTPL}{\ensuremath{{\tt BTPL}}}

\newcommand{\reg}{\ensuremath{{\tt Reg}}}
\newcommand{\regt}{\ensuremath{r}}

\newcommand{\buyers}{buyer's~}
\newcommand{\Buyers}{Buyer's~}

\newcommand{\Po}{\ensuremath{\mathcal{P}}}
\newcommand{\G}{\ensuremath{\mathcal{G}}}

\begin{document}

\title{Learning in Auctions: Regret is Hard, Envy is Easy
}

\author{
Constantinos Daskalakis\thanks{Supported by a Microsoft Research Faculty Fellowship, and NSF Award CCF-0953960 (CAREER) and CCF-1551875. This work was done in part while the author was visiting the Simons Institute for the Theory of Computing.}\\
EECS, MIT\\
{\tt costis@csail.mit.edu}
\and
Vasilis Syrgkanis\thanks{This work was done in part while the author was visiting the Simons Institute for the Theory of Computing.}\\
Microsoft Research, NYC\\\
{\tt vasy@microsoft.com}
}
\date{}

\maketitle
\thispagestyle{empty}
\addtocounter{page}{-1}

\begin{abstract}
A large line of recent work studies the welfare guarantees of simple and prevalent combinatorial auction formats, such as selling $m$ items via simultaneous second price auctions (SiSPAs)~\cite{Christodoulou2008,Bhawalkar2011,Feldman2013}. These guarantees hold even when the auctions are repeatedly executed and the players use no-regret learning algorithms to choose their actions. Unfortunately, off-the-shelf no-regret learning algorithms for these auctions are computationally inefficient as the number of actions available to each player is exponential. We show that this obstacle is insurmountable: there are no polynomial-time no-regret learning algorithms for SiSPAs, unless $\RP\supseteq \NP$, even when the bidders are unit-demand. Our lower bound raises the question of how good outcomes polynomially-bounded bidders may discover in such auctions.

To answer this question, we propose a novel concept of learning in auctions, termed ``no-envy learning.'' This notion is founded upon Walrasian equilibrium, and we show that it is both efficiently implementable and results in approximately optimal welfare, even when the bidders have valuations from the broad class of fractionally subadditive (XOS) valuations (assuming demand oracle access to the valuations) or coverage valuations (even without demand oracles). No-envy learning outcomes are a relaxation of no-regret learning outcomes, which maintain their approximate welfare optimality while endowing them with computational tractability. Our result for XOS valuations can be viewed as the first instantiation of approximate welfare maximization in combinatorial auctions with XOS valuations, where both the designer and the agents are computationally bounded and agents are strategic. Our positive and negative results extend to many other simple auction formats that have been studied in the literature via the smoothness paradigm. 

Our positive results for XOS valuations are enabled by a novel Follow-The-Perturbed-Leader algorithm for settings where the number of experts and states of nature are both infinite, and the payoff function of the learner is non-linear. We show that this algorithm has applications outside of auction settings, establishing big gains in a recent application of no-regret learning in security games. Our efficient learning result for coverage valuations is based on a novel use of convex rounding schemes and a reduction to online convex optimization. 
\end{abstract}
\newpage
\section{Introduction} \label{sec:intro}
A central challenge in Algorithmic Mechanism Design is to {understand the effectiveness and limitations of mechanisms to induce economically efficient outcomes in a computationally efficient manner.} A practically relevant and most actively studied setting for performing this investigation is that of {\em combinatorial auctions}. 

This setting involves a seller with a set $[m]$ of indivisible items, which he wishes to sell to a set $[n]$ of buyers. Each buyer $i \in [n]$ is characterized by a valuation function $v_i:2^{[m]} \rightarrow \mathbb{R}_+$, assumed monotone, which maps each bundle $S_i$ of items to the buyer's value $v_i (S_i)$ for this bundle. This function is known to the buyer, but is unknown to the seller and the other buyers. The seller's goal is to find a partition $S_1\sqcup S_2 \sqcup \ldots \sqcup S_n=[m]$ of the items together with prices $p_1,\ldots,p_n$ so as to maximize the total welfare resulting from allocating bundle $S_i$ to each buyer $i$ and charging him $p_i$. The total buyer utility from such an allocation would be $\sum_{i}(v_i (S_i )-p_i)$ and the seller's revenue would be $\sum_i p_i$, so the total welfare from such an allocation would simply be $\sum_i v_i (S_i)$. 

Given the seller's uncertainty about the buyer's valuations, she needs to interact with them to select a good allocation. However, the buyers are strategic, aiming to optimize their own utility, $v_i(S_i )-p_i$. Hence, the seller needs to design her allocation and price computation rules carefully so that a good allocation is found despite the agents' strategization in response to these rules. How much of the optimal welfare can the seller guarantee?
 
A remarkable result in Economics is that welfare can be exactly optimized, as long as we have unbounded computational and communication resources, via the {celebrated VCG mechanism~\cite{Vickrey1961,Clarke1971,Groves1973}.} This mechanism asks bidders to report their valuations, uses their reports at face value to select an optimal partition of the items, and computes payments in a way that it is in the best interest of all bidders to truthfully report their valuations; in particular, it is a {\em dominant strategy truthful} mechanism, and because of its truthfulness it guarantees that an optimal allocation is truly selected. 

Despite its optimality and truthfulness, the VCG mechanism is overly demanding in terms of both computation and communication. Reporting the whole valuation functions is too expensive for the bidders to do for most interesting types of valuations. Moreover, optimizing welfare exactly with respect to the reported valuations is also difficult in many cases. Unfortunately, if we are only able to do it approximately, the truthfulness of the VCG mechanism disappears, and no welfare guarantees can be made. 
{Even with computational concerns set aside, it is widely acknowledged that the VCG mechanism is rarely used in practice \cite{Ausubel2006}. At the same time, many practical scenarios involve the allocation of items through simple mechanisms which are often} not centrally designed and non-truthful. Take eBay, for example, where several different items are sold simultaneously and sequentially via ascending price and other types of auctions. Or consider sponsored search where several keywords are auctioned simultaneously and sequentially using generalized second price auctions. For most interesting families of valuations such environments induce non truthful behavior, and are thus difficult to study analytically.

{The prevalence of such simple decentralized auction environments provides motivation for a quantitative analysis of the quality of outcomes in simple non-truthful mechanisms.} A growing volume of research has taken up this challenge, developing tools for studying the welfare guarantees of non-truthful mechanisms; see e.g.~\cite{Bikhchandani1999,Christodoulou2008,Bhawalkar2011,Hassidim2011,Fu2012,Syrgkanis2013,Feldman2013}. Using the approximation perspective, this literature bounds the Price-of-Anarchy (PoA) of simple non-truthful mechanisms, and has provided remarkable insights into their economic efficiency. 

To illustrate these results, let us consider Simultaneous Second Price Auctions, which we will abbreviate to ``SiSPAs'' in the remainder of this paper. While we focus our attention on these auctions, our results extend to the most common other forms of auctions studied in the PoA literature; see Section~\ref{sec:extension to smooth mechanisms} for a discussion. As implied by its name, a SiSPA asks every bidder to bid on each of the items separately and allocates each item using a second price auction based on the bids submitted solely for this item. 

Facing a SiSPA, a bidder whose valuation is non-additive is not able to express his  complex preferences over bundles of items. It is thus a priori not clear how he will bid, and what the  resulting welfare will be. One situation where a prediction can be made is when the bidders have some information about each other, either knowing each other's valuations, or knowing a distribution from which each others valuations are drawn. In this case, we can study the SiSPA's Nash or Bayesian Nash equilibrium behavior, computing the welfare in equilibrium. Remarkably, the work on the PoA of mechanisms has shown that the equilibrium welfare of SiSPAs (and of other types of simple auctions) is guaranteed to be within a constant factor of optimum, even when the bidders' valuations are subadditive \cite{Feldman2013}.\footnote{A {\em subadditive} valuation $v$ is one satisfying $v(S \cup T) \le v(S)+v(T),$ for all $S, T \subseteq [m]$.} When bidders have no information about each other, the problem becomes ill-posed, as it is impossible for the bidders to form beliefs about each others bids in order to choose their own bid.

A way out of the conundrum comes from the realization that simple mechanisms often occur repeatedly, involving the same set of bidders; think sponsored search.
In such a setting it is natural to assume that bidders engage in learning to compute their new bids as a function of their experience so far. One of the most standard types of learning behavior 
is that of {\em no-regret learning}. A bidder's bids  over $T$ executions of a SiSPA satisfy the no-regret learning guarantee if the bidder's cumulative utility over the $T$ executions is within an additive $o(T)$ of the cumulative utility that the bidder would have achieved from the best in hindsight vector of bids $b_1,\ldots,b_m$, if he were to place the same bid $b_j$ on item $j$ in all $T$ executions of the SiSPA. Assuming that bidders use no-regret learning  to update their bids in repeated executions of a SiSPA (or other types of simple auctions) the afore-referenced work has shown that the average bidder welfare across the $T$ executions is within a constant factor of the otpimal welfare, even when the bidders' valuations are subadditive \cite{Feldman2013}.

These guarantees are astounding, especially given the intractability results for dominant strategy truthful mechanisms, which hold even when the bidders have {submodular valuations \cite{Dobzinski2011,Dobzinski2012,Dughmi2015}---a family of valuations that is smaller than subadditive.}\footnote{A {\em submodular} valuation $v$ is one satisfying $v(S\cup T) + v(S \cap T) \le v(S)+v(T)$, for all $S, T \subseteq [m]$.} However, moving to simple non-truthful auctions does not come without a cost. Cai and Papadimitriou \cite{Cai2014} have recently established intractability results for computing Bayesian-Nash equilibria in SiSPAs, even for quite simple types of valuations, namely mixtures of additive and unit-demand~\cite{Cai2014}.\footnote{A {\em unit-demand} valuation $v$ is one satisfying $v(S)=\max_{i \in S} v(\{i\})$, for all $S \subseteq [m]$.} 
At the same time, implementing no-regret learning in combinatorial auctions is quite tricky as the action space of the bidders explodes. For example, in SiSPAs there is a continuum of possible bid vectors that a bidder may submit and, even if we tried to discretize this set, their number would typically be exponential in the number of items in order to maintain a good approximation from the discretization. {Unfortunately, no-regret algorithms typically require in every step computation that is linear in the number of available actions, hence in our case exponential in the number of items.} 

{An important open question in the literature  has thus been whether this obstacle can be overcome via specialized no-regret algorithms that only need polynomial computation. Our first result shows that this obstacle is insurmountable.} We show that in one of the most basic settings where no-regret learning is non-trivial, it cannot be implemented in polynomial-time unless {\tt RP} $\supseteq$ {\tt NP}.

\begin{theorem} \label{thm:hardness of unit-demand learning}
Suppose that a unit-demand bidder whose value for each item $i\in [m]$ is $v$ participates in $T$ executions of a SiSPA. Unless {\tt RP} $\supseteq$ {\tt NP}, there is no learning algorithm running in time polynomial in $m$, $v$, and $T$ and whose regret is any polynomial in $m$, $v$, and $1/T$. The computational hardness holds even when the learner faces i.i.d. samples from a fixed distribution of competing bids, and whether or not no-overbidding is required of the bids produced by the learner.
\end{theorem}

Note that our theorem proves an intractability result even if pseudo-polynomial dependence on the description of $v$ is permitted in the regret bound and the running time. The {\em no-overbidding assumption} mentioned in the statement of our theorem represents a collection of conditions under which no-regret learning in second-price auctions gives good welfare guarantees~\cite{Christodoulou2008,Feldman2013}. An example of such no-overbidding condition is this: For each subset $S \subseteq [m]$, the sum of bids across items in $S$ does not exceed the bidder's value for bundle $S$. Sometimes this condition is only required to hold on average. It will be clear that our hardness easily applies whether or not no-overbidding is imposed on the learner, so we do not dwell on this issue more in this paper. 

{How can we show the in-existence of computationally efficient no-regret learning algorithms? A crucial (and general) connection that we establish in this paper is that it suffices to prove an inapproximability result for a corresponding offline combinatorial optimization problem. More precisely, we prove Theorem~\ref{thm:hardness of unit-demand learning} by establishing an inapproximability result for an offline optimization problem related to SiSPAs, together with a ``transfer theorem'' that transfers in-approximability from the offline problem to intractability for the online problem. The transfer theorem is a generic statement applicable to any online learning setting.} In particular, we show the following; see Section~\ref{sec:hardness} for details.
\begin{enumerate}
\item In SiSPAs, finding the best response payoff against a polynomial-size supported distribution of opponent bids  is strongly {\tt NP}-hard to additively approximate for a unit-demand bidder. Another way to say this is that  one step of a specific learning algorithm, namely Follow-The-Leader (FTL), is inapproximable. See Theorems~\ref{thm:one step of FTL is hard} and~\ref{thm:approx-optimal}.

\item In any setting where finding an optimum for an explicitly given distribution of functions over some set ${\cal F}$ is hard to additively approximate, no efficient no-regret learner against sequences of functions from ${\cal F}$ exists, unless $\RP\supseteq \NP$. This result is generic, saying that whenever one step of FTL is inapproximable, there is no no-regret learner. See Theorem~\ref{thm:LB for FTL implies general LB}. 
\end{enumerate}


The intractability result of Theorem~\ref{thm:hardness of unit-demand learning} casts shadow in the ability of computationally bounded learners to achieve no-regret guarantees in combinatorial auctions where their action space explodes with the number of items {and the number of items is large}. We have shown this for SiSPAs, but our techniques easily extend to Simultaneous First Price Auctions, and we expect to several other commonly studied mechanisms for which PoA bounds are known. With the absence of efficiently implementable learning algorithms, it is unclear when we should expect computationally bounded bidders to actually converge to approximately efficient outcomes in these auctions.

From a design standpoint it may be interesting to identify conditions for the bidder valuations and the format of the auction under which no-regret learning is both efficiently implementable and leads to approximately optimal outcomes. While this direction is certainly interesting, it would not address the question of what welfare we should expect of SiSPAs and other simple auctions that have been studied in the literature, or how much of the PoA bounds can be salvaged for computationally bounded bidders. Moreover, recent results of Braverman et al. \cite{Braverman2016} show that for a large class of auction schemes where no-regret algorithms are efficiently computable, no-better than a logarithmic in the number of items welfare guarantee can be achieved (which is achievable by the single-bid auction of \cite{Devanur2015}).   

{We propose an alternative approach to obtaining robust welfare guarantees of simple auctions for computationally bounded players by 
introducing a new type of learning dynamics, which we call \emph{no-envy}, and which are founded upon the concept of Walrasian equilibrium.} 
{In all our results, no-envy learning outcomes are a super-set of no-regret learning outcomes. We show that this super-set simultaneously achieves two important properties: i) while being a broader set, it still maintains the welfare guarantees of the set of no-regret learning outcomes established via PoA analyses; ii) there exist computationally efficient no-envy learning algorithms; when these algorithms are used by the bidders, their joint behavior converges (in a decentralized manner) to the set of no-envy learning outcomes for a large class of valuations (which includes submodular). Thus no-envy learning provides a way to overcome the computational intractability of no-regret learning in auctions with implicitly given exponential action spaces.} We describe our results in the following section. We will focus our attention on SiSPAs but the definition of no-envy learning naturally extends to any mechanism and all our positive results extend to a large class of smooth mechanisms; see Section~\ref{sec:extension to smooth mechanisms}. 

\begin{figure}
\begin{center}
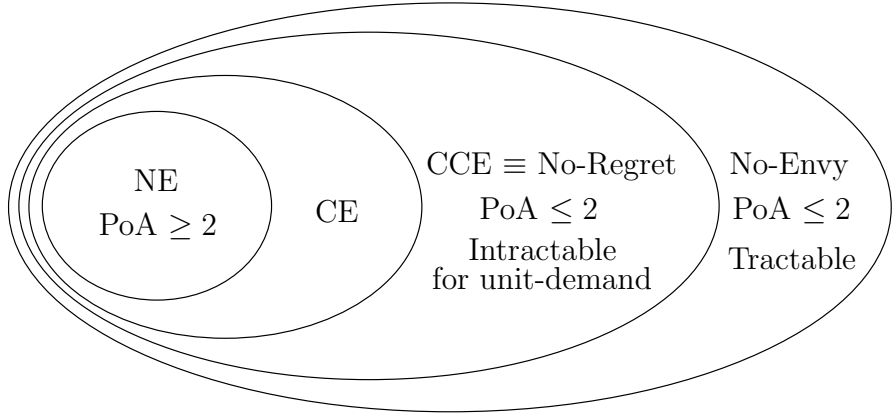
\end{center}
\caption{\textit{\small{We depict the state of the world for simultaneous second price auctions with XOS bidders. NE denotes the set of Nash equilibria, CE the set of correlated equilibria, CCE the set of coarse correlated equilibria which are equivalent to the limit empirical distributions of no-regret dynamics. Last with No-Envy we denote the limit empirical distributions of no-envy dynamics. PoA refers to the ratio of the optimal welfare over the worst welfare achieved by any solution in each set. Tractability in this figure refers to the existence of polynomial time decentralized algorithms that each player can invoke in the agnostic setting and converge to the no-regret or no-envy condition at a polynomial error rate.}}}
\end{figure}

\subsection{No-Envy Dynamics: Computation and Welfare.} No-envy dynamics is a twist to no-regret dynamics. Recall that in no-regret dynamics the requirement is that the cumulative utility of the bidder after $T$ rounds be within an additive $o(T)$ error of the optimum utility he would have achieved had he played the best fixed bid in hindsight. In no-envy dynamics, we require that the bidder's cumulative utility be within an additive $o(T)$ of the optimum utility that he would have achieved if he was allocated the best in hindsight fixed bundle of items in all rounds and paid the price of this bundle in each round. The guarantee is inspired by Walrasian equilibrium: In auctions, the prices that a bidder faces on each bundle of items is determined by the bids of the other bidders. Viewed as a price-taker, the bidder would want to achieve utility at least as large as the one he would have achieved if he purchased his favorite bundle at its price. No-envy dynamics require that the average utility of the bidder across $T$ rounds is within $o(1)$ of what he would have achieved by purchasing the optimal bundle at its average price in hindsight. 

Inspired by Walrasian equilibrium, no-envy learning defines a natural benchmark against which to evaluate an online sequence of bids. It is easy to see that in SiSPAs the no-regret learning requirement is stronger than the no-envy learning requirement. Indeed, the no-envy requirement is implied by the no-regret requirement against a subset of all possible bid vectors, namely those in $\{0,+\infty\}^m$. So no-envy learning is more permissive than no-regret learning, allowing for a broader set of outcomes. This not true necessarily for other auction formats, but it holds for the types of valuation functions and auctions studied in this paper; (see proof of Lemma~\ref{lem:from sets to bids} and Definition \ref{defn:value-covering}). In particular, in all our no-envy learning upper bounds the set of outcomes reachable via no-envy dynamics is always a superset of the outcomes reachable via no-regret dynamics. Moreover, this is true even if the no-envy dynamics are constrained to not overbid. 

To summarize, for all types of valuations and auction formats studied in this paper, no-envy learning is a relaxation of no-regret learning, permitting a broader set of outcomes. While no-regret learning outcomes are intractable, we show that this broader set of outcomes is tractable. At the same time, we show that this broader set of outcomes maintains approximate welfare optimality. So we have increased the set of possible outcomes, but maintained their economic efficiency and endowed them with computational efficiency.

We proceed to describe our results for the computational and economic efficiency of no-envy learning. Before proceeding, we should point out that, while in our world the no-envy learning guarantee is a relaxation of the no-regret learning guarantee, the problem of implementing no-envy learning sequences remains similarly challenging. Take SiSPAs, for example. As we have noted no-envy learning is tantamount to requiring the bidder to not have regret against all bid vectors in $\{0,+\infty\}^m$. This set is exponential in the number of items $m$, so it is unclear how to run an off-the-shelf no-regret learner efficiently. In particular, we are still suffering from the combinatorial explosion in the number of actions, which lead to our lower bound of Theorem~\ref{thm:hardness of unit-demand learning}. Yet the curse of dimensionality is now much more benign. Our upper bounds, discussed next, establish that we can harness big computational savings when we move from competing against any bid vector to competing against bid vectors in $\{0,+\infty\}^m$. Except to do this we still need to develop new general-purpose, no-regret algorithms for online learning settings where the number of experts is exponentially large and the {cost/utility} functions are arbitrary.

\subsubsection{Efficient No-Envy Learning.} We show that no-envy learning can be efficiently attained for bidders with fractionally subadditive (XOS) valuations. A valuation $v(\cdot)$  belongs to this family if for some collection of vectors ${\cal V}=(v^\ell)_\ell$, where each $v^\ell \in \mathbb{R}_+^m$, it satisfies:
\begin{align} \label{eqn:XOS}
v(S)=\max_{v^\ell \in {\cal V}}\sum_{j \in S}v^{\ell}_j, \forall S \subseteq [m].
\end{align}
Note that the XOS class is larger than that of submodular valuations. In many applications, the set ${\cal V}$ describing an XOS valuation may be large. Thus instead of inputting this set explicitly into our algorithms, we will assume that we are given an oracle, which given $S$ returns the vector $v^\ell \in {\cal V}$ such that $v(S) =  \sum_{j \in S}v^{\ell}_j$. Such an oracle is known as an {\em XOS oracle} \cite{Dobzinski2006,Feige2006}. We will also sometimes assume, as it is customary in Walrasian equilibrium, that we are given access to a {\em demand oracle}, which given a price vector $p \in \mathbb{R}_+^m$ returns the bundle $S$ maximizing $v(S)-\sum_{j\in S}p_j$. We show the following.

\begin{theorem}\label{thm:XOS-demand}
Consider a bidder with an XOS valuation $v(\cdot)$ participating in a sequence of SiSPAs. Assuming access to a demand and an XOS oracle for $v(\cdot)$,\footnote{For submodular valuations this is equivalent to assuming access only to demand oracles, as XOS oracles can be simulated in polynomial time assuming demand oracles \cite{Dobzinski2010}} there exists a polynomial-time algorithm for computing the bidder's bid vector $b^t$ at every time step $t$ such that after $T$ iterations the bidder's average utility satisfies:
\begin{align}
\frac{1}{T}\E\left[\sum_{t=1}^{T} u(b^t)\right] \geq \max_{S} \left(v(S) - \sum_{j\in S} \hat{\theta}_j^T\right)-{O\left(\frac{m^2(D+H)}{\sqrt{T}}\right)}, \label{eq:den xasame}
\end{align}
where $\hat{\theta}_j^T$ is the average cost of item $j$ in the $T$ executions of the SiSPA as defined by the bids of the competing bidders, $D$ is an upper bound on the {competing bid} for any item and $H$ is an upper bound on $\max_S v(S)$. The learning algorithm with the above guarantee also satisfies the no overbidding condition that the sum of bids for any set of items is never larger than the bidder's value for that set. Moreover, the guarantee holds with no assumption about the behavior of competing bidders. Finally, extensions of this algorithm to other smooth mechanisms are provided in Section~\ref{sec:extension to smooth mechanisms}.
\end{theorem}

The proof of Theorem~\ref{thm:XOS-demand} is carried out in three steps, of which the first and last are specific to SiSPAs, while the second provides a general-purpose Follow-The-Perturbed-Leader (FTPL) algorithm in online learning settings where the number of experts is exponentially large and the {cost/utility} functions are arbitrary:
\begin{enumerate}
\item The first ingredient is simple, using the XOS oracle to reduce no-envy learning in SiSPAs to no-regret learning in a related ``online buyer's problem,'' where the learner's actions are not bid vectors but instead what bundle to purchase, prior to seeing the prices; see Definition~\ref{def:online buyer's prob}. Theorem~\ref{lem:from sets to bids} provides the reduction from no-envy learning to this problem using XOS oracles. {This reduction can also be done (albeit starting from approximate no-envy learning) for several mechanisms that have been analyzed through the smoothness framework of \cite{Syrgkanis2013} as we elaborate in Section~\ref{sec:extension to smooth mechanisms}.}

\item The second step proposes a FTPL algorithm for general online learning problems where the learner chooses some action $a^t \in A$ and the environment chooses some state $\theta^t \in \Theta$, {\em from possibly infinite, unstructured sets $A$ and $\Theta$}, and where the learner's reward is tied to these choices through some function $u(a^t,\theta^t)$ {\em that need not be linear}. Since $A$ need not have finite-dimensional representation and $u$ need not be linear, we cannot efficiently perturb (either explicitly or implicitly) the cumulative rewards of the elements in $A$ as required in each step of FTPL~\cite{Kalai2005}; see~\cite{Bubeck2012} and its references for an overview of such approaches. Instead of perturbing the cumulative rewards of actions in $A$ directly, our proposal is to do this indirectly by augmenting the history $\theta^1,\ldots,\theta^{t-1}$ that the learner has experienced so far with some randomly chosen fake history, and run Follow-The-Leader (FTL) subject to these augmentations. While it is not a priori clear whether our perturbation approach is a useful one, it is clear that our proposed algorithm only needs an offline optimization oracle to be implemented, as each step is an FTL step after the fake history is added. {When applying this algorithm to the online buyer's problem from Step 1, the required offline optimization oracle will conveniently end up being simply a demand oracle.}

Our proposed general purpose learner is presented in Section~\ref{sec:oracles}. The way our learner accesses function $u$ is via an optimization oracle, which given a finite multiset of elements from $\Theta$ outputs an action in $A$ that is optimal against the uniform distribution over the multiset. See Definition~\ref{sec:optimization oracle}. In Theorem~\ref{thm:general-regret}, we bound the regret experienced by our algorithm in terms of $u$'s stability. Roughly speaking, the goal of our randomized augmentations of the history in each step of our learning algorithm is to smear the output of the optimization oracle applied to the augmented sequence over $A$, allowing us to couple the choices of Be-The-Perturbed-Leader and Follow-The-Pertubed-Leader.

\item  To apply our general purpose algorithm from Theorem \ref{thm:general-regret} to the online buyer's problem for SiSPAs from Step 1, we need to bound the stability of the bidder's utility function subject to a good choice of a history augmentation sampler. This is done in Section~\ref{sec:demand}. There turns out to be a simple sampler for our application here, where only one price vector is added to the history, whose prices are independently distributed according to an exponential distribution with mean $O(\sqrt{T})$ and variance $O(T)$.

\item While our motivation comes from mechanism design, our FTPL algorithm from Step 2 is general purpose, and we believe it will find applications in other settings. We provide some relevant discussion in Section~\ref{sec:finite-parameter}, where we show how our algorithm implies regret bounds independent of $|A|$ when $\Theta$ is finite, as well as quantitative improvements on the regret bounds of a recent paper of Balacn et al. for security games~\cite{Balcan2015}. 
\end{enumerate}

In the absence of demand oracles, we provide positive results for the subclass of XOS called coverage valuations.
To explain these valuations, consider a bidder with $k$ needs, $1,\ldots,k$, associated with values $a_1,\ldots,a_k$. There are $m$ available items, each covering a subset of these needs. So we can view each item as a set $\beta_i \subseteq [k]$ of the needs it satisfies. The value that the bidder derives from a set $S \subseteq [m]$ of the items is the total value from the needs that are covered, namely: 
\begin{equation}\label{eqn:coverage-defn}
v(S)=\sum_{\ell \in \cup_{j} \beta_j} a_\ell.
\end{equation}

\begin{theorem}\label{thm:main-coverage}
Consider a bidder with an explicitly given coverage valuation $v(\cdot)$ participating in a sequence of SiSPAs. There exists a polynomial-time algorithm for computing the bidder's bid vector $b^t$ at every time step $t$ such that after $T$ iterations the 
bidder's utility satisfies:
\begin{align}
\frac{1}{T}\E\left[\sum_{t=1}^{T} u(b^t)\right] \geq \max_{S} \left(\left(1-\frac{1}{e}\right)v(S) - \sum_{j\in S} \hat{\theta}_j^T\right) - 3m\frac{H+\sqrt{D}}{\sqrt{T}}, \label{eq:kati xasame}
\end{align}
where $\hat{\theta}_j^T$, $H$ and $D$ are as in Theorem~\ref{thm:XOS-demand}, and the algorithm satisfies the same no overbidding condition stated in that theorem. There is no assumption about the behavior of the competing bidders, and extensions of this algorithm to other smooth mechanisms are provided in Section~\ref{sec:extension to smooth mechanisms}.
%
\end{theorem}
\noindent Notice that our no-envy guarantee~\eqref{eq:kati xasame} in Theorem~\ref{thm:main-coverage} has incurred a loss of a factor of $1-1/e$ in front of $v(S)$, compared to the no-envy guarantee~\eqref{eq:den xasame}. This relaxed guarantee is an even broader relaxation of the no-regret guarantee. Still, as we show in the next section this does not affect our approximate welfare guarantees. We prove Theorem~\ref{thm:main-coverage} via an interesting connection between the online buyer's problem for coverage valuations and the convex rounding approach for truthful welfare maximization proposed by~\cite{Dughmi2011}. In the online buyer's problem, recall that the buyer needs to decide what set to buy at each step, prior to seeing the prices. It is natural to have the buyer include each item to his set independently, thereby defining an expert for all points $x\in [0,1]^m$, where $x_i$ is the probability that item $i$ is included. It turns out that the expected utility of the buyer under such distribution $x$ is not necessarily convex, so this choice of experts turns our online learning problem non-convex. Instead we propose to massage each expert $x \in [0,1]^m$ into a distribution $D(x) \in \Delta(2^{[m]})$ and run online learning on the massaged experts. In Definition~\ref{def:convex rounding} we put forth conditions for the massaging operation $D(\cdot)$ under which online learning becomes convex and gives approximate no-regret (Lemma~\ref{lem:online-convex}). We then instantiate $D(\cdot)$ with the Poisson sampling of~\cite{Dughmi2011}, establishing Theorem~\ref{thm:main-coverage}. Our approach is summarized in Section~\ref{sec:convex} and the details can be found in Appendix~\ref{sec:app-convex}.

\subsubsection{Welfare Maximization.}  Arguably one of the holy grails in Algorithmic Mechanism Design, since its inception, has been to obtain polynomial-time mechanisms optimizing welfare in incomplete information settings. 
{We show that SiSPAs achieve constant factor approximation welfare guarantees for the broad class of XOS valuations at every no-envy or approximate no-envy learning outcome. 
Thus the relaxation from no-regret to no-envy learning does not degrade the quality of the welfare guarantees, and has the added benefit that no-envy outcomes can be attained by computationally bounded players in a decentralized manner, using our results from the previous section. In Section~\ref{sec:extension to smooth mechanisms}, we show that this property applies to a large class of mechanisms that have been analyzed in the literature via the \emph{smoothness} paradigm \cite{Syrgkanis2013}.}

\begin{corollary} \label{cor:welfare guarantees}
When each bidder $i \in \{1,\ldots,n\}$ participating in a sequence of SiSPAs has an XOS valuation (endowed with a demand and XOS oracle) or an explicitly given coverage valuation $v_i(\cdot)$, there exists a polynomial-time computable learning algorithm such that, if each bidder $i$ employs this algorithm to compute his bids $b_i^t$ at each step $t$, then after $T$ rounds the average welfare is guaranteed to be at least:
\begin{equation}
\frac{1}{T}\E\left[\sum_{t=1}^{T} SW(b_1^t,\ldots,b_n^t)\right] \geq \frac{1}{2} \left(1-\frac{1}{e}\right)\opt(v_1,\ldots,v_n) -O\left( {m^2}\cdot {n \cdot} \max_{S,i}v_i(S) \sqrt{\frac{1}{T}}\right).
\end{equation}
If all bidders have XOS valuations with demand and XOS oracles the factor in front of OPT is $1/2$.
\end{corollary}

We regard Corollary~\ref{cor:welfare guarantees}, in particular our result for XOS valuations with demand queries, as alleviating the  intractability of no-regret learning in simple auctions.
It also provides a new perspective to mechanism design, namely mechanism design with no-envy bidders.
In doing so, it proposes an answer to the question raised by~\cite{Feldman2015} about whether demand oracles can be exploited for welfare maximization with submodular bidders. We show a positive answer for the bigger class of XOS valuations, albeit with a different solution concept. (It still remains open whether there exist poly-time dominant strategy truthful mechanisms for submodular bidders with demand queries.) We believe that no-envy learning is a fruitful new approach to mechanism design, discussing in Section~\ref{sec:extension to smooth mechanisms} the meaning of the solution concept outside of SiSPAs.

\section{Preliminaries}

We analyze the online learning problem that a bidder faces when participating in a sequence of repeated executions of a simultaneous second price auction (SiSPA) with $m$ items. While we focus on SiSPAs our results extend to the most commonly studied formats of simple auctions, as discussed in Section~\ref{sec:extension to smooth mechanisms}. A sequence of repeated executions of a SiSPA corresponds to a sequence of repeated executions of a game involving $n$ players (bidders).
At each execution $t$, each player $i$ submits a bid $b_{ij}^t$ on each item $j$. We denote by $b_i^t$ the vector of bidder $i$'s bids at time $t$ and by $b^t$ the profile of bids of all players on all items. Given these bids, each item is given to the bidder who bids for it the most and this bidder pays the second highest bid on the item. Ties are broken according to some arbitrary tie-breaking rule. 
Each player $i$ has some fixed (across executions) valuation $v_i: 2^{[m]}\rightarrow \R_+$ over bundles of items. If at time $t$ he ends up winning a set of items $S^t$ and is asked to pay a price of $\theta_j^t$ for each item $j\in S^t$, then his utility is $v_i(S^t) - \sum_{j\in S^t}\theta_j^t$, i.e. his utility is assumed quasi-linear. An important class of valuations that we will consider in this paper is that of XOS valuations, defined in Equation \eqref{eqn:XOS}, which are a super-set of submodular valuations but a subset of subadditive valuations. We will also consider the class of {coverage valuations}, defined in Equation~\ref{eqn:coverage-defn}, which are a subset of XOS. Different results will consider different types of access to an XOS valuation through an XOS oracle, a demand oracle, or a value oracle, as described in the introduction. For more properties of these oracles see \cite{Dobzinski2010}.
%

\paragraph{Online bidding problem.} From the perspective of a single player $i$, all that matters to him to calculate his utility in a SiSPA is the highest  bid submitted by the other bidders on each item $j$, as well as the probability that he wins each item $j$ if he ties with the highest other bid on that item. 
For simplicity of notation, we will assume throughout the paper that the player always loses an item when he ties first. All our results, both positive and negative, easily extend to the more general case of arbitrary bid-profile dependent tie-breaking. 
Since we will analyze learning from the perspective of a single player, we will drop the index of player $i$. 
For a fixed bid profile of the opponents, we will refer to the highest other bid on item $j$ as the threshold of item $j$ and denote it with $\theta_{j}$. We also denote with $\theta=\left(\theta_{1},\ldots,\theta_{m}\right)$. 
The player wins an item $j$ if he submits a bid $b_{j}>\theta_j$ and loses the item otherwise. When he wins item $j$, he pays $\theta_j$.
We are interested in learning algorithms 
that achieve a no-regret guarantee even when the thresholds of the items are decided as it is customary by an adversary. Thus, the online learning problem that 
a player faces in a simultaneous second price auction is defined as follows:
\begin{defn}[Online bidding problem]
At each execution/day/time/step $t$, the player picks a bid vector $b^t$ and the adversary picks adaptively (based on the history of the player's past bid vectors but not on the bidder's current bid vector $b^t$) a threshold vector $\theta^t$. The player wins the set $S(b^t,\theta^t)=\{j\in [m]: b_j^t>\theta_j^t\}$ and gets reward:
\begin{equation}
u(b^t,\theta^t) = v\left(S\left(b^t,\theta^t\right)\right) - \sum_{j\in S\left(b^t,\theta^t\right)} \theta_j^t.
\end{equation}
\end{defn}

We allow a learning algorithm to be randomized, i.e. submit a random bid vector at each step whose distribution may depend on the history of past threshold vectors. We will evaluate a learning algorithm based on its \emph{regret} against the best fixed bid vector in hindsight. 

\begin{defn}[Regret of Learning Algorithm]
The expected average regret of a randomized online learning algorithm against a sequence $\theta^{1:T}=(\theta^1,\ldots,\theta^T)$ of threshold vectors is: 
\begin{equation}
\reg\left(\theta^{1:T}\right)=\E_{b^{1:T}}\left[\sup_{b^*}\frac{1}{T}\sum_{t=1}^T \left(u(b^*,\theta^t) - u(b^t,\theta^t)\right)\right],
\end{equation}
where recall that $b^t$ is random and depends on $\theta^{1:t-1}$, as specified by the online learning algorithm.
The regret $\regt(T)$ of the algorithm against an adaptive adversary is the maximum regret against any adaptively chosen sequence of $T$ threshold vectors. An algorithm has polynomial regret rate if $\regt(T)=\poly(T^{-1},m,\max_S v(S))$.\end{defn}

\section{Hardness of No-Regret Learning}\label{sec:no regret hardness}\label{sec:hardness}


We will show that there does not exist an efficiently computable learning algorithm with polynomial regret rate for the {online bidding problem} for SiSPAs unless {\tt RP} $\supseteq$ {\tt NP}, proving a proof of Theorem~\ref{thm:hardness of unit-demand learning}. We first examine a related offline optimization problem which we show is $\NP$-hard to approximate to within a small additive error. We then show how this inapproximability result implies the non-existence of polynomial-time no-regret learning algorithms for SiSPAs unless {\tt RP} $\supseteq$ {\tt NP}.
Throughout this section we will consider the following very restricted class of valuations: the player is unit-demand and has a value $v$ for getting any item, i.e. his value for any set of items is given by $v(S) = v\cdot 1\{S\neq \emptyset\}$. Our intractability results are strong intractability results in the sense that they hold even if we assume that $v$ is provided in the input in unary representation.

\paragraph{Optimal Bidding Against An Explicit Threshold Distribution is Hard.} \label{sec:inapproximability of best response}
We consider the following optimization problem: 
\begin{defn}[Optimal Bidding {Problem}] A distribution $D$ of threshold vectors $\theta$ over a set of $m$ items is given explicitly as a list of $k$ {vectors}, where $D$ is assumed to choose a uniformly random vector from the list. A bidder has {a unit-demand valuation} with the same value $v$ for {each} item, given in unary. The problem asks for {a bid vector} that maximizes the bidder's expected utility against distribution $D$. In fact, it only asks to compute the expected value from an optimal bid vector, i.e.
\begin{equation}
\sup_{b} \E_{\theta\sim D}\left[u(b,\theta)\right] = \sup_{b} \left\{ v\cdot \Pr\left[\exists j\in [m]: b_j>\theta_j\right] - \sum_{j\in [m]} \theta_j\cdot \Pr\left[b_j>\theta_j\right]\right\}.
\end{equation}
\end{defn}

We show that the {optimal bidding problem} is $\NP$-hard via a reduction from $r$-regular set-cover. In fact we show that it is hard to approximate, up to an additive approximation that is inverse-polynomially related to the input size. This will be useful when using the hardness of this problem to deduce the in-existence of efficiently computable learning algorithms with polynomial regret rates.

\begin{theorem}[Hardness of Approximately Optimal Bidding]\label{thm:approx-optimal}\label{thm:exact-optimal}
The optimal bidding problem is $\NP$-hard to approximate to within an additive $\xi$ even when: the $k$ threshold vectors in the support of {(the explicitly given distribution)} $D$ take values in $\{1,H\}^m$, $H=k^2\cdot m^2$, $v=2\cdot k\cdot m$ and $\xi={1 \over 2k}$.
\end{theorem}

An interesting interpretation Theorem~\ref{thm:approx-optimal} is that the Follow-The-Leader (FTL) algorithm is intractable in SiSPAs for unit-demand bidders.  Indeed, every step of FTL needs to find a bid vector that is a best response to the empirical distribution of the threshold vectors that have been encountered so far. See Theorem~\ref{thm:one step of FTL is hard} in Appendix~\ref{sec:app-hardness} and the discussion around this theorem.

\paragraph{Efficient No-Regret implies Poly-time Approximately Optimal Bidding.} \label{sec:from inapproximability of best response to NP-hardness of no-regret}
	Given the hardness of \emph{optimal bidding} in SiSPAs, we are ready to sketch the proof of our main impossibility result (Theorem~\ref{thm:hardness of unit-demand learning}) for {\em online bidding} in SiSPAs. Our result holds even if the possible threshold vectors that the bidder may see take values in some known discrete finite set. It also holds even if we weaken the regret requirements of the {online bidding problem}, only requiring that the player achieves no-regret with respect to bids of the form $\{0,v/2m\}^m$, i.e., the bid on each item is either $0$ or an $2m$-th faction of the player's value. Notice that any such bid is a non-overbidding bid. 
Hence, the no-regret requirement that we impose is weaker than achieving no-regret against any fixed bid/any fixed no-overbidding bid. We will refer to the afore-described weaker learning task as the \emph{simplified online bidding problem}. We sketch here how to deduce from the inapproximability of optimal bidding the impossibility of polynomial-time no-regret learning (even for the simplified online bidding problem), deferring full details to Appendix~\ref{sec:app:final proof of lower bound}.

\begin{proofskof}{Theorem~\ref{thm:hardness of unit-demand learning}} We present the structure of our proof and the challenges that arise, leaving details for Section~\ref{sec:app:final proof of lower bound}. 
Consider a hard distribution $D$ for the optimal bidding problem from Theorem~\ref{thm:approx-optimal}, and let $b^*$ be the bid vector that optimizes the expected utility of the bidder when a threshold vector is drawn from $D$. Also, let $u^*$ be the corresponding optimal expected utility. (Theorem~\ref{thm:approx-optimal} says that approximating $u^*$ is {\tt NP}-hard.) Now let us draw $T$ i.i.d. samples $\theta=(\theta^1,\ldots,\theta^T)$ from $D$.  Clearly, if $T$ is large enough, then, with high probability, the expected utility $\tilde{u}_\theta$ of $b^*$ against the uniform distribution over $\theta^1,\ldots,\theta^T$ is approximately equal to $u^*$.

Now let us present the sequence $\theta^1,\ldots,\theta^T$ to a no-regret learning algorithm. The learning algorithm is potentially randomized so let us call $\hat{u}_\theta$ the expected average utility (over the randomness in the algorithm and keeping sequence $\theta$ fixed) that the algorithm achieves when facing the sequence of threshold vectors $\theta^1,\ldots,\theta^T$. If the regret of the algorithm is $r(T)$, this means that $\hat{u}_\theta \ge \sup_b({1 \over T}\sum_{t=1}^Tu(b,\theta^t)) - r(T) \ge \tilde{u}_\theta-r(T)$. In particular, if $r(T)$ scales polynomially with $1/T$ then, for large enough $T$, $\hat{u}_\theta$ is lower bounded by $\tilde{u}_\theta$ (minus some small error), and hence by $u^*$ (minus some small error). Hence, $\hat{u}_\theta$ (plus some small error) provides an upper bound to $u^*$. Moreover, if we run our no-regret learning algorithm a large enough number of times $N$ against the same sequence of threshold vectors and average the average utility achieved by the algorithm in these $N$ executions, we can get a very good estimate of $\hat{u}_\theta$, and hence a very good upper bound for $u^*$, with high probability. The paragraph ``Upper Bound'' in Appendix~\ref{sec:app:final proof of lower bound} gives the details of this part.

The challenge that we need to overcome now is that, in principle, the expected average utility $\hat{u}_\theta$ of our no-regret learner against sequence $\theta^1,\ldots,\theta^T$ could be much larger than $\sup_b({1 \over T}\sum_{t=1}^Tu(b,\theta^t))$ and hence $\tilde{u}_\theta$ and $u^*$, as the algorithm is allowed to change its bid vector in every step. We need to argue that this cannot happen. In particular, we would like to upper bound $\hat{u}_\theta$ by $u^*$. We do this via a Martingale argument exploiting the randomness in the choice of the sequence $\theta$. Using Azuma's inequality, we show that for large enough $T$, the $\hat{u}_\theta$ is upper bounded by $u^*$ plus some small error with high probability. In fact we show something stronger: if $T$ is large enough then, with high probability, $u^*$ plus some small error upper bounds the algorithm's average utility (not just average expected utility), where now both the threshold and the bid vectors are left random. Hence, we can argue that, with high probability, if we run our algorithm $N$ times over a (long enough) sequence of random threshold vectors and we compute the average (across the $N$ executions) of the average (across the $T$ steps) utility of our algorithm, then this double average is upper bounded by $u^*$ plus some small error. Hence, we get a lower bound on $u^*$. (One execution would indeed suffice, but we need to argue about the average across $N$ executions given the way we obtain our upper bound in the previous paragraph.) The paragraph ``Lower Bound'' in Appendix~\ref{sec:app:final proof of lower bound} gives the details of this part.

Overall, if we choose $T$ and $N$ large enough polynomials in the description of the hard instance of the optimal bidding problem from Theorem~\ref{thm:approx-optimal}, then all approximation errors can be made arbitrary inverse polynomials, providing any desired (inverse polynomial) approximation to the optimal utility $u^*$ against distribution $D$, with high probability. Since getting an inverse polynomial approximation is an $\NP$-hard problem, this implies that there cannot exist a polynomial-time no-regret learning algorithm with polynomial regret rate, {unless} $\RP\supseteq \NP$. 
\end{proofskof}


\section{Walrasian Equilibria and No-Envy Learning in Auctions}\label{sec:no-envy}

The hardness of no-regret learning in simultaneous auctions motivates the investigation of other notions of learning that have rational foundations and at the same time admit efficient implementations. Our inspiration in this paper comes from the study of markets and the well-studied notion of {\em Walrasian equilibrium}. Recall that an allocation of items to buyers together with a price on each item constitutes a Walrasian equilibrium if no buyer envies some other allocation at the current prices. That is the bundle $S$ allocated to each buyer maximizes the difference $v(S)-p(S)$ of his value $v(S)$ for the bundle minus the cost of the bundle. Implicitly the Walrasian equilibrium postulates some degree of rationality on the buyers: given the prices of the items, each buyer wants a bundle of items such that  he has \emph{no-envy} against getting any other bundle at the current prices. 

We adapt this \emph{no-envy} requirement to SiSPAs (and other mechanisms in Section~\ref{sec:extension to smooth mechanisms}). In a SiSPA a player is facing a set of prices on the items, which are determined by the bids of the other players and are hence unknown to him when he is choosing his bid vector. In a sequence of repeated executions of a SiSPA, the player needs to choose a bid vector at every time-step. The fact that he does not know the realizations of the item prices when making his choice turns the problem into a learning problem. We will say that the sequence of actions that he took satisfies the no-envy guarantee, if in the long run he does not regret not buying any fixed set $S$ at its average price.

\begin{defn}[No-Envy Learning Algorithm]
An algorithm for the online bidding problem is a {\em no-envy algorithm}
if, for any adaptively chosen sequence of threshold vectors $\theta^{1:T}$ by an adversary, the bid vectors $b^1,\ldots,b^T$ chosen by the algorithm satisfy:
\begin{equation}
\E\left[\frac{1}{T}\sum_{t=1}^T u(b^t,\theta^t)\right] \geq \max_{S\subseteq [m]} \left(v(S) - \sum_{j\in S}\hat{\theta}_j^T \right)-\epsilon(T)
\end{equation}
where $\hat{\theta}_j^T=\frac{1}{T}\sum_{t=1}^{T} \theta_j^t$ and $\epsilon(T)\rightarrow 0$. It has polynomial envy rate if $\epsilon(T) = \poly(T^{-1},m,\max_S v(S))$.
\end{defn}

To allow for even larger classes of settings to have efficiently computable no-envy learning outcomes, we will also define a relaxed notion of no-envy. In this notion the player is guaranteed that his utility is at least some $\alpha$-fraction of his value for any set $S$, less the average price of that set. The latter is a more reasonable relaxation in the online learning setting given that, unlike in a market setting, the players do not know the realization of the prices when they make their decision.

\begin{defn}[Approximate No-Envy Learning Algorithm]
An algorithm for the online bidding problem is an $\alpha$-approximate no-envy algorithm if, for any adaptively chosen sequence of threshold vectors $\theta^{1:T}$ by an adversary, the bid vectors $b^1,\ldots,b^T$ chosen by the algorithm satisfy:
\begin{equation}
\E\left[\frac{1}{T}\sum_{t=1}^T u(b^t,\theta^t)\right] \geq\max_{S\subseteq [m]} \left(\frac{1}{\alpha}v(S) - \sum_{j\in S}\hat{\theta}_j^T \right)-\epsilon(T).
\end{equation}
\end{defn}

To gain some intuition about the difference between no-envy and no-regret learning guarantees consider the following. When we compute the utility from a fixed bid vector in hindsight, then in every iteration the set of items that the player would have won is nicely correlated with that round's threshold vector 
in the sense that the player wins an item in that round only when the item's threshold is low. On the contrary, when evaluating the player's utility had he won a specific set of items in all rounds the player may win and pay for an item even when the price of the item is high. The results of this section imply that for XOS valuations, the no-regret condition is stronger than the no-envy condition. Hence, when we analyze no-envy learning algorithms for XOS bidders we relax the algorithm's benchmark. Correspondingly, if the bidders of a SiSPA are XOS and use no-envy learning algorithms to update their bid vectors, the set of outcomes that they may converge to is broader than the set of no-regret outcomes. So, in comparison to no-regret learning outcomes, our positive results in this section pertain to a broader set of outcomes, endowing them with computational tractability and as we will see also approximate welfare optimality. 

\paragraph{Roadmap.} In the rest of this section we  reduce the no-envy learning problem to a related online learning problem, which we call the \emph{online \buyers problem}. We show that achieving \emph{no-envy} in the \emph{online bidding problem} can be reduced to achieving \emph{no-regret} in the \emph{online \buyers problem}. Similarly, achieving $\alpha$-approximate no-envy can be reduced to some form of approximate no-regret. Lastly we show that \emph{no-envy learning} implies good welfare: if all players in the simultaneous second-price auction game follow a no-envy learning algorithm then the average welfare of the selected allocations is approximately optimal. In subsequent sections we will provide efficiently computable no-envy or approximate no-envy algorithms for the \emph{online \buyers problem}. Finally, our positive results extend to the most commonly studied mechanisms through the smoothness framework, as we elaborate in Section~\ref{sec:extension to smooth mechanisms}. 
 
\subsection{Online \Buyers Problem}\label{sec:sets-to-bids}\label{sec:buyers}  We first show that we can reduce the no-envy learning problem to a related online learning problem, which we call the \emph{online \buyers problem}.
\begin{defn}[Online \buyers problem] \label{def:online buyer's prob}
Imagine a buyer with some valuation $v(\cdot)$ over a set of $m$ items who is asked to request a subset of the items to buy each day before seeing their prices. In particular, at each time-step $t$ an adversary picks a set of thresholds/prices $\theta_{j}^t$ for each item $j$ adaptively based on the past actions of the buyer. Without observing the thresholds at step $t$, the buyer picks a set $S^t$ of items to buy. His instantaneous reward is:
\begin{equation}\label{eqn:buyer-utility}
u(S^t,\theta^t) = v(S^t)- \sum_{j\in S^t}\theta_{j}^t,
\end{equation}
i.e., the buyer receives the set $S^t$ and pays the price for each item in the set.
\end{defn}

For simplicity, we overload notation and denote by $u(b,\theta)$ the reward in the online bidding problem from a bid vector $b$ and with $u(S,\theta)$ the reward in the online \buyers problem from a set $S$. We relate the online \buyers problem to the online bidding problem in SiSPAs in a black-box way, by showing that when the valuations are XOS, then any algorithm which achieves no-regret or ``approximate'' no-regret for the \emph{online \buyers problem} can be turned in a black-box and efficient manner into a no-envy algorithm for the \emph{online bidding problem}, assuming access to an XOS oracle.

\begin{lemma}[From buyer to bidder] \label{lem:from sets to bids} Suppose that we are given access to an efficient learning algorithm for the online \buyers problem which guarantees for any adaptive adversary: 
\begin{equation}\label{eqn:alg-buyers-guarantee}
\E\left[\frac{1}{T} \sum_{t=1}^T u(S^t,\theta^t)\right] \geq \max_{S} \left( \frac{1}{\alpha}v(S) - \sum_{j\in S} \hat{\theta}_j^T\right) - \epsilon(T),
\end{equation}
where $\hat{\theta}_j^T=\frac{1}{T}\sum_{t=1}^{T} \theta_j^t$. Then we can construct an efficient $\alpha$-approximate no-envy  algorithm for the online bidding problem, assuming access to XOS value oracles. Moreover, this algorithm never submits an overbidding bid. 
\end{lemma}


\paragraph{A trivial example: efficient no-envy for $O(\log(m))$-capacitated XOS.} Consider a buyer with a $d$-capacitated XOS valuation, i.e. the valuation is XOS and for any set $S$: $v(S) = \max_{T\subseteq S: |T|\leq d} v(T)$. If $d=O(\log(m))$, then 
it suffices for the buyer to achieve no-regret  against sets of size $d$, which are $2^d$. This is polynomial if $d=O(\log(m))$. Thus we can simply invoke any off-the-shelf no-regret learning algorithm, such as multiplicative weight updates \cite{Auer1995}, where each set of $d$ is treated as an expert, and apply it to the \emph{online \buyers problem}. This would be efficiently computable and would lead to a regret rate of $O(\sqrt{T\log(2^d)}= O(\sqrt{T\cdot d})$. By Lemma~\ref{lem:from sets to bids}, we then get an efficiently computable exact no-envy algorithm with the same envy rate. 

\medskip \noindent The challenge addressed by our paper is to remove the bound on $d$, which we address in the next sections.

\subsection{No-Envy Implies Approximately Optimal Welfare}
We conclude 
by showing that if all players in a SiSPA 
use an $\alpha$-approximate no-envy learning algorithm, then the average welfare is a $2\alpha$-approximation to the optimal welfare, less an additive error term stemming from the envy of the players. In other words the price of anarchy of $\alpha$-approximate no-envy dynamics is upper bounded by $2\alpha$.

\begin{theorem}\label{thm:envy-poa}
If $n$ players participating in repeated executions of a SiSPA use an $\alpha$-approximate no-envy learning algorithm with envy rate $\epsilon(T)$ and which does not overbid, then in $T$ executions of the SiSPA the average bidder welfare is at least $\frac{1}{2\alpha}\opt - n\cdot \epsilon(T)$, where $\opt$ is the optimal welfare for the input valuation profile $v=(v_1,\ldots,v_n)$.
\end{theorem}

\section{Online Learning with Oracles}\label{sec:oracles}

In this section we devise novel follow-the-perturbed leader style algorithms for general online learning problems. We then apply these algorithms and their analysis to get no-envy learning algorithms (Section~\ref{sec:demand})  for the online bidding problem. In Section \ref{sec:app-finite} we instantiate our analysis to learning problems where the adversary can only pick one among finitely many parameters and give implications of this setting to  no-regret learning algorithms (Section~\ref{sec:app-finite}) for the online bidding problem, with a finite number of possible thereshold vectors. In Section~\ref{sec:app-security}, we also give implications to security games \cite{Balcan2015}. 

Consider an online learning problem where at each time-step an adversary picks a parameter $\theta^t\in \Theta$ and the algorithm picks an action $a^t\in A$. The algorithm receives a reward: $u(a^t,\theta^t)$, which could be positive or negative. We will assume that the rewards are uniformly bounded by some function of the parameter $\theta$, for any action $a\in A$, i.e.: $\forall a\in A: u(a,\theta) \in \left[-f_-(\theta),f_+(\theta)\right]$. 
We will denote with $\theta^{1:t}$ a sequence of parameters $\{\theta_1,\theta_2,\ldots,\theta_t\}$. Moreover, we denote with: $
U(a,\theta^{1:t}) = \sum_{\tau=1}^t u(a,\theta^\tau)$, 
the cumulative utility of a fixed action $a\in A$ for a sequence of choices $\theta^{1:t}$ of the adversary. 

\begin{defn}[Optimization oracle] \label{sec:optimization oracle} We will consider the case where we are given oracle access to the following optimization problem: given a sequence of parameters $\theta^{1:t}$ compute some optimal action for this sequence: 
\begin{equation}
\Or{\theta^{1:t}} = \arg\max_{a\in A} U(a,\theta^{1:t}).
\end{equation}
\end{defn}

We define a new type of perturbed leader algorithms where the perturbation is introduced in the form of extra samples of parameters:
\begin{algorithm}[Follow the perturbed leader with sample perturbations]\label{defn:ftpl}
At each time-step $t$:
\begin{enumerate}
 \item Draw a random sequence of parameters $\{x\}^t=\{x^1,\ldots,\ldots,x^k\}^t$ independently and based on some time-independent distribution over sequences. Both the length of the sequence and the parameter $x^i\in \Theta$ at each iteration of the sequence can be random.  
 \item Denote with $\{x\}^t\cup \theta^{1:t-1}$ the augmented sequence of parameters where we append the extra parameter samples $\{x\}^t$ at the beginning of sequence $\theta^{1:t-1}$
 \item Invoke oracle $M$ and play action:
\begin{equation}
a^t = \Or{\{x\}^t\cup \theta^{1:{t-1}}}.
\end{equation}
 \end{enumerate}
\end{algorithm}

Using a reduction of \cite{Hutter2005} (see their Lemma 12) we can show that to bound the regret of Algorithm~\ref{defn:ftpl} against adaptive adversaries it suffices to bound the regret against oblivious adversaries (who pick the sequence non-adaptively), of the following algorithm, which only draws the samples once ahead of time (see Appendix~\ref{sec:app-oblivious}). In subsequent sections, we analyze this algorithm and setting.
%

\begin{algorithm}[Follow the perturbed leader with fixed sample perturbations]\label{defn:lazy-ftpl}
Draw a random sequence of parameters $\{x\}=\{x^1,\ldots,\ldots,x^k\}$ based on some distribution over sequences and at the beginning of time. At each time-step $t$, invoke oracle $M$ and play action: $
a^t = \Or{\{x\}\cup \theta^{1:{t-1}}}$.
\end{algorithm}

\paragraph{Perturbed Leader Regret Analysis.}
We give a general theorem on the regret of a perturbed leader algorithm with sample perturbations. In the sections that follow we will give instances of this analysis in two online learning settings related to no-envy and no-regret dynamics in our bidding problem and provide concrete regret bounds. 

\begin{theorem}\label{thm:general-regret}
Suppose that the distribution over sample sequences $\{x\}$, satisfies the \emph{stability property} that for any sequence of parameters $\theta^{1:T}$ and for any $t\in [1:T]$:
\begin{equation}
\E_{\{x\}}\left[u(\Or{\{x\}\cup \theta^{1:t}},\theta^t)-u( \Or{\{x\}\cup \theta^{1:{t-1}}}, \theta^t)\right]\leq  g(t)
\end{equation}
Then the expected regret of Algorithm~\ref{defn:lazy-ftpl} against oblivious adversaries is upper bounded by:
\begin{equation}
\sup_{a^*\in A}\E_{\{x\}}\left[\sum_{t=1}^T \left( u(a^*,\theta^t) - u(a^t,\theta^t)\right)\right] \leq \sum_{t=1}^T g(t) +  \E_{\{x\}}\left[\sum_{x^{\tau}\in \{x\}} \left(f_-(x^\tau)+f_+(x^\tau)\right)\right]
\end{equation}
Hence, the regret of Algorithm~\ref{defn:ftpl} against adaptive adversaries is bounded by the same amount. 
\end{theorem}

\subsection{Efficient No-Envy Learning with Demand Oracles}\label{sec:demand}

We will apply the perturbed leader approach to the online buyer's problem we defined in Section~\ref{sec:buyers}. Then using Lemma~\ref{lem:from sets to bids} we can turn any such algorithm to a no-envy learning algorithm for the original bidding problem in second price auctions, when the valuations fall into the XOS class. 

In the online buyer's problem the action space is the collection of sets $A=2^m$, while the parameter set of the adversary is to pick a threshold $\theta_j$ for each item $j$, i.e. $\Theta=\R_+^m$. The reward $u(S,\theta)$, at each round from picking a set $S$, if the adversary picks a vector $\theta\in \Theta$ is given by Equation~\eqref{eqn:buyer-utility}. We will instantiate Algorithm~\ref{defn:lazy-ftpl} for this problem and apply the generic approach of the previous section. We will specify the exact distribution over sample sequences that we will use and we will bound the functions $f_-(\cdot)$, $f_+(\cdot)$ and $g(\cdot)$. First, observe that the reward is bounded by a function of the threshold vector:
$u(S,\theta)\in \left[-\|\theta\|_1,H\right]$, 
where $H$ is an upper bound on the valuation function, i.e. $v([m])<H$.

\paragraph{Optimization oracle.} 
It is easy to see that the offline problem for a sequence of parameters $\theta^{1:t}$ is exactly a \emph{demand oracle}, where the price on each item $j$ is its average threshold $\hat{\theta}_j^t$ in hindsight. 

\paragraph{Single-sample exponential perturbation.} We will use the following sample perturbation: we will only add one sample $x\in \Theta$, where the coordinate $x_i$ of the sample is distributed independently and according to an exponential distribution with parameter $\epsilon$, i.e. for any $k\geq 0$ the density of $x_i$ at $k$ is $f(k)=\frac{1}{2}\epsilon e^{-\epsilon k}$, while it is $0$ for $k<0$.

The most important part of the analysis is proving a stability bound for our algorithm. We provide such a proof in Appendix~\ref{sec:app-demand-thm}. Given the stability bound we then apply Theorem~\ref{thm:general-regret} to get a bound for Algorithm~\ref{defn:lazy-ftpl} with a single sample exponential perturbation.
\begin{theorem}\label{thm:demand-regret-bound}
Algorithm~\ref{defn:lazy-ftpl} when applied to the online buyers problem with a single-sample exponential perturbation with parameter $\epsilon = \sqrt{\frac{1}{H D T}}$, where $D$ is the maximum threshold that the adversary can pick and $H$ is the maximum value, runs in randomized polynomial time, assuming a demand oracle and achieves regret:
\begin{equation*}
\sup_{a\in A} \sum_{t=1}^T \E\left[u(a,\theta^t)-u(a^t,\theta^t)\right]\leq  2(mD+H) m (\log(T) + 1)+4m \sqrt{(mD+H)DT} = O\left(m^2(D+H)\sqrt{T}\right)
\end{equation*}
\end{theorem}

Theorem~\ref{thm:demand-regret-bound}, Lemma~\ref{lem:from sets to bids} and the reduction form oblivious to adaptive adversaries, imply a polynomial time no-envy algorithm for the online bidding problem assuming access to demand and XOS oracles. If valuations are submodular, then XOS oracles can be simulated in polynomial time via demand oracles \cite{Dobzinski2010}, thereby only requiring access to demand oracles. Thus we get Theorem~\ref{thm:XOS-demand}.

\section{Efficient No-Envy Learning via Convex Rounding}
\label{sec:convex}

In this section we show how to design efficient approximate no-envy learning algorithms via the use of the convex rounding technique, which has been used in approximation algorithms and in truthful mechanism design, and via online convex optimization applied to an appropriately defined online learning problem in a relaxed convex space. Though our techniques can be phrased more generally, throughout the section we will mostly cope with the concrete case where the valuation of the player is an explicitly given coverage valuation. These valuations have been well-studied in combinatorial auctions \cite{Dughmi2011} and are a subset of submodular valuations. Answering value and XOS queries for such valuations can be done in polynomial time \cite{Dughmi2011,Dobzinski2010}.
 \begin{defn}[Coverage valuation] A coverage valuation is given via the means of a vertex-weighted hyper-graph $\G=(V,E)$. Each item $j\in [m]$ corresponds to a hyper-edge. Each vertex $v\in V$ has a weight $w_v\geq 0$. The value of the player for a set $S$ is the sum of the vertices of the hyper-graph, that is contained in the union of the hyper-edges corresponding to the items in $S$.
\end{defn}

\paragraph{Proving Theorem~\ref{thm:main-coverage}.} Based on Lemma~\ref{lem:from sets to bids}, in order to design an $\alpha$-approximate no-envy algorithm for the online bidding problem, it suffices to design an efficient algorithm for the online \buyers problem with guarantees as described in Lemma~\ref{lem:from sets to bids}. In the remainder of the section we will design such an algorithm for the \emph{online \buyers problem} with $\alpha=\frac{e}{e-1}$ and for explicit coverage valuations, thereby proving Theorem~\ref{thm:main-coverage}. Subsequently, by Theorem~\ref{thm:envy-poa} the latter will imply a price of anarchy guarantee of $\frac{2e}{e-1}$ for such dynamics.
The only missing piece in the proof of Theorem~\ref{thm:main-coverage} is the following lemma, whose full proof we defer to Appendix~\ref{sec:app-convex}.
\begin{lemma}\label{lem:buyers-learning}
If the bidder's valuation $v(\cdot)$ is an explicitly given coverage valuation, there exists a polynomial-time computable learning algorithm for the online \buyers problem that guarantees that for any adaptively chosen sequence of thresholds $\theta^{1:T}$ with $\theta_j^t\leq K$:
\begin{equation}
\E\left[\frac{1}{T} \sum_{t=1}^T u(S^t,\theta^t)\right] \geq \max_{S} \left(\left(1-\frac{1}{e}\right) v(S) - \sum_{j\in S} \hat{\theta}_j^T\right)- 3m\frac{\max_{j\in [m]}v(\{j\})+\sqrt{K}}{\sqrt{T}},
\end{equation}
\end{lemma}
\begin{proofsk}
Suppose that the buyer picks a set at each iteration at random from a distribution where each item $j$ is included independently with probability $x_j$ to the set. Then for any vector $x$, the expected utility of the buyer from such a choice is $\E_{S^t \sim x^t}\left[u(S^t,\theta^t)\right] = V(x^t) - \langle\theta^t, x^t\rangle$,
where $V(\cdot)$ is the multi-linear extension of $v(\cdot)$ and $\langle x, y\rangle$ is the inner product between vectors $x$ and $y$. If $V(\cdot)$ was concave we could invoke online convex optimization algorithms, such as the projected gradient descent of \cite{Zinkevich2003} and get a regret bound, which would imply a regret bound for the buyers problem. However, $V(\cdot)$ is not concave for most valuation classes. We will instead use a \emph{convex rounding scheme}, which is a mapping from any vector $x$ to a distribution over sets $D(x)$ such that $F(x) = \E_{S\sim D(x)}\left[v(S)\right]$ is a concave function of $x$. We also require that the marginal probability of each item be at most the original probability of that item in $x$. If the rounding scheme satisfies that for any integral $x$ associated with set $S$, $F(x)\geq \frac{1}{\alpha} v(S)$, then we can call an online convex optimization algorithm on the concave function $F(x)-\langle \theta,x\rangle$. Then we show that this yields an $\alpha$-approximate no-envy algorithm for the online buyers problem.
\end{proofsk}

\section{No-Envy Learning for General Mechanisms}\label{sec:extension to smooth mechanisms}

In this section we generalize our approach to most smooth mechanisms \cite{Syrgkanis2013} that have been analyzed in the literature. For ease of exposition we only focus on mechanisms for combinatorial auction settings, even though the approach could be employed for more general mechanism design settings. 

A general mechanism $M$ for a combinatorial auction setting is defined via an action space $\A_i$ available to each player $i$, an allocation function, which maps each action profile $a\in \A\triangleq\A_1\times\ldots\times\A_n$ to a feasible partition $X(a)$ of the items among players, as well as a payment function, which maps an action profile to a vector of payments for $P(a)$ for each player. We denote with $X_i(a)$ and $P_i(a)$ the allocation and payment of player $i$. These functions could also output randomized allocations and payments, but for simplicity of notation we restrict to deterministic mechanisms.

First and foremost we need to generalize the definition of no-envy to general mechanisms other than simultaneous second price auctions. To achieve this we need to define the equivalent of a threshold vector for a general mechanism. We define the notion of a threshold-payment for a player $i$ and a set $S$, which will coincide with the sum of thresholds - $\sum_{j\in S}\theta_{ij}$ -  for the case of a simultaneous second price auction.  

\begin{defn}[Threshold Payment]
Given a set $S$ and an action profile $a$, the threshold payment for player $i$ for set $S$ is the minimum payment he needs to make to win set $S$, i.e.:
\begin{equation}
\tau_i(S,a_{-i}) = \inf_{a_i'\in \A_i: X_i(a_i',a_{-i}) \supseteq S} P_i(a_i',a_{-i})
\end{equation}
The threshold function is additive if:
\begin{equation}
\tau_i(S,a_{-i}) = \sum_{j\in S} \theta_{ij}(a_{-i})
\end{equation}
for some item specific functions $\theta_{ij}$, derived based on the auction rules.
\end{defn}

The average threshold payment for a set $S$, takes the role of the average price of the set, in a repeated learning environment. Thus we can analogously define a no-envy learning algorithm for any repeated mechanism setting, where mechanisms $M$ is repeated over time among the same players for $T$ iterations and at each iteration each player picks an action $a_i^t\in \A_i$.

\begin{defn}[No-Envy Learning for General Mechanisms]
An algorithm for a repeated mechanism setting is an $\alpha$-approximate no-envy algorithm if for any adaptively and adversarially chosen sequence of opponent actions $a_{-i}^{1:T}$:
\begin{equation}
\E\left[\frac{1}{T}\sum_{t=1}^T u(a_i^t,a_{-i}^t)\right] \geq \max_{S\subseteq [m]} \left(\frac{1}{\alpha}v(S) - \frac{1}{T}\sum_{t\in T} \tau_i(S,a_{-i}^t) \right)-\epsilon(T)
\end{equation}
where $\epsilon(T)\rightarrow 0$. It has polynomial envy rate if $\epsilon(T) = \poly(T^{-1},m,\max_S v(S))$.
\end{defn}

\paragraph{Sufficient conditions on the mechanism.} We now give conditions on the mechanism $M$, such that it admits efficient no-envy learning dynamics and such that any approximate no-envy outcome is also approximately efficient. Our conditions can be viewed as a stronger version of the smooth mechanism definition of Syrgkanis and Tardos \cite{Syrgkanis2013}, as well as a generalization of the value and revenue covering formulation of Hartline et al. \cite{Hartline2014}.

We begin by reminding the reader of the definition of a smooth mechanism \cite{Syrgkanis2013} specialized to a combinatorial auction setting.
\begin{defn}[\cite{Syrgkanis2013}]
A mechanism is $(\lambda,\mu)$-smooth if for any action profile $a\in \A$, there exists for each player $i$ an action $a_i^*$ for each player $i$, such that:
\begin{equation}
\sum_{i\in [n]} u_i(a_i^*,a_{-i}) \geq \lambda \opt - \mu \rev(a)
\end{equation}
where $\rev(a)=\sum_{i\in [n]} P_i(a)$ is the revenue of the auctioneer and $\opt$ is the optimal welfare.
\end{defn}

To apply our approach we will refine the smoothness definition and require a stronger ``smoothness'' property, albeit one that holds for almost all mechanisms that have been analyzed via the smooth mechanism framework. Our stronger smoothness version is more inline with the revenue and value covering framework of \cite{Hartline2014} and can be thought of as an ex-post version of that framework. However, unlike the approach in \cite{Hartline2014} our definition applies to general multi-dimensional mechanism design environments.  

We will follow the terminology of \cite{Hartline2014} of revenue and value covering. Our definition is stronger than the smooth mechanism definition in two ways. First it requires a deviation inequality for each individual player, rather than on aggregate across players. Moreover, it requires a smoothness inequality not only for the optimal allocation but rather we would require one for every possible allocation. In that respect it is closer to the solution-based smoothness of \cite{Lykouris2016} and to the original definition of smooth games of \cite{Roughgarden2009}. All of these strengthenings seem essential for our approach on designing no-envy dynamics to work. 

Now we are ready to present the definitions of ex-post value and threshold covering, which are a stronger version of the smoothness definition. 
\begin{defn}[Ex-post $\lambda$-value covered]\label{defn:value-covering}
A mechanism is ex-post $\lambda$-value covered if for any feasible allocation profile $x=(S_1,\ldots,S_n)$, there exists for each player $i$ an action $a_i^*(S_i)\in \A_i$ such that for any action profile $a\in \A$:
\begin{equation}
u_i(a_i^*(S_i),a_{-i})+\tau_i(S_i,a_{-i})\geq \lambda v_i(S_i) 
\end{equation}
\end{defn}
\begin{defn}[Ex-post $(\mu_1,\mu_2)$-threshold covered]
A mechanism is ex-post $(\mu_1,\mu_2)$-threshold covered if for any action profile $a\in \A$ and allocation profile $x=(S_1,\ldots,S_n)$:
\begin{equation}
\sum_{i\in [n]} \tau_i(S_i,a_{-i}) \leq \mu_1 \rev(a) + \mu_2 SW(a)
\end{equation}
where $SW(a) = \sum_{i\in [n]} v_i(X_i(a))$.
\end{defn}
It is easy to see that if a mechanism is $\lambda$-value covered and $(\mu,0)$-threshold covered, then it is $(\lambda,\mu)$-smooth according to \cite{Syrgkanis2013}. We add the extra welfare term, to enable the analysis of second-price auctions too. This term is related to the weakly $(\lambda,\mu_1,\mu_2)$-smooth mechanisms in \cite{Syrgkanis2013}.

\paragraph{No-envy learning and welfare.} Now we are ready to give the generalizations of our main theorems for general mechanisms. First we argue that if a mechanism is $(\mu_1,\mu_2)$-threshold covered and players use no-envy learning, then the average welfare is approximately optimal. The proof of this theorem follows along very similar lines as in the proof of Theorem \ref{thm:envy-poa} and hence we omit the proof. 
\begin{theorem}[No-Envy Welfare for General Mechanisms]
If a mechanism is ex-post $(\mu_1,\mu_2)$-threshold covered and each player invokes an $\alpha$-approximate no-envy algorithm with envy rate $\epsilon(T)$, then after $T$ iterations the average welfare in the auction is at least $\frac{1}{\alpha(\max\{1,\mu_1\}+\mu_2)}\opt - n\cdot \epsilon(T)$, where $\opt$ is the optimal welfare for the input valuation profile $v=(v_1,\ldots,v_n)$.
\end{theorem}

\paragraph{Efficient no-envy algorithms.} Next we argue that if a mechanism is $\lambda$-value covered, then the existence of an efficient no-envy learning algorithm reduces to the existence of an efficient no-regret algorithm for the natural generalization of the online \buyers problem. 

\begin{defn}[Online \buyers problem for general mechanisms]
A buyer with some valuation $v(\cdot)$ over a set of $m$ items wants to decide on each day which items to buy. At each time-step $t$ an adversary picks an opponent action profile $a_{-i}^t$ adaptively based on the past actions of the buyer. Without observing $a_{-i}^t$ at step $t$, the buyer picks a set $S^t$ to buy. His reward is:
\begin{equation}\label{eqn:general-buyer-utility}
u(S^t,a_{-i}^t) = v(S^t)- \tau_i(S^t,a_{-i}^t),
\end{equation}
i.e., the buyer receives the set $S^t$ and pays the threshold price for the set.
\end{defn}

\begin{lemma}[From buyer to bidder in general mechanisms]\label{lem:buyer-bidder-general} Suppose that the mechanism is ex-post $\lambda$-value covered and that we are given access to an efficient learning algorithm for the online \buyers problem which guarantees for any adaptive adversary: 
\begin{equation}
\E\left[\frac{1}{T} \sum_{t=1}^T u(S^t,a_{-i}^t)\right] \geq \max_{S\subseteq [m]} \left( \frac{1}{\alpha}v(S) - \frac{1}{T}\sum_{t=1}^T\tau_i(S^t,a_{-i}^t)\right) - \epsilon(T)
\end{equation}
Then we can construct an efficient $\frac{\alpha}{\lambda}$-approximate no-envy  algorithm for the online bidding problem, assuming access to XOS value oracles. 
\end{lemma}
The proof of the latter Lemma follows along very similar lines as in the proof of Lemma \ref{lem:from sets to bids}, hence we omit its proof.

Last it is easy to see that when the threshold functions are additive, then the online buyer's problem for general mechanisms is exactly the same as the online buyer's problem for the simultaneous second price auction mechanism. Thus our results in the main sections of the paper, provide an efficient algorithm for the online buyer's problem with $\alpha=1$ for XOS valuations assuming access to a demand and XOS oracle and with $\alpha=(1-1/e)$ for coverage valuations assuming access to a value oracle.

\begin{theorem}\label{thm:general-XOS-demand}
Consider a bidder with an XOS valuation $v(\cdot)$ participating in $\lambda$-value covered mechanism with additive threshold functions. Assuming access to a demand and an XOS oracle for $v(\cdot)$, there exists a polynomial-time algorithm for computing the bidder's action  $a_i^t$ at every time step $t$ such that after $T$ iterations the bidder's average utility satisfies:
\begin{align}
\frac{1}{T}\E\left[\sum_{t=1}^{T} u(a^t)\right] \geq \max_{S} \left(\lambda\cdot v(S) - \frac{1}{T}\sum_{t=1}^T\tau_i(S^t,a_{-i}^t)\right)-\vsedit{O\left(\frac{m^2(D+H)}{\sqrt{T}}\right)}, \end{align}
where $D$ is an upper bound on the threshold function $\theta_{ij}(\cdot)$ for any item and $H$ is an upper bound on $\max_S v(S)$. The guarantee holds with no assumption about the behavior of competing bidders.
\end{theorem}

\begin{theorem}
Consider a bidder with an explicitly given coverage valuation $v(\cdot)$ participating in $\lambda$-value covered mechanism with additive threshold functions. There exists a polynomial-time algorithm for computing the bidder's action $a_i^t$ at every time step $t$ such that after $T$ iterations the 
bidder's utility satisfies:
\begin{align}
\frac{1}{T}\E\left[\sum_{t=1}^{T} u(a^t)\right] \geq \max_{S} \left(\lambda \cdot \left(1-\frac{1}{e}\right) \cdot v(S) - \frac{1}{T}\sum_{t=1}^T\tau_i(S^t,a_{-i}^t)\right) - 3m\frac{H+\sqrt{D}}{\sqrt{T}}, 
\end{align}
$H$ and $D$ are as in Theorem~\ref{thm:general-XOS-demand}. There is no assumption about the behavior of the competing bidders, 
%
\end{theorem}

\paragraph{Main result for general mechanisms.} Combining the aforementioned discussion and analysis we can draw the following main conclusion of this section.
\begin{corollary} 
When each bidder $i \in \{1,\ldots,n\}$ participating in a sequence of ex-post $\lambda$-value covered and $(\mu_1,\mu_2)$-threshold covered mechanisms with additive threshold functions, has an XOS valuation (endowed with a demand and XOS oracle) or an explicitly given coverage valuation $v_i(\cdot)$, there exists a polynomial-time computable learning algorithm such that, if each bidder $i$ employs this algorithm to compute his action $a_i^t$ at each step $t$, then after $T$ rounds the average welfare is guaranteed to be at least:
\begin{equation*}
\frac{1}{T}\E\left[\sum_{t=1}^{T} SW(a^t)\right] \geq \frac{\lambda}{\max\{1,\mu_1\}+\mu_2}\left(1-\frac{1}{e}\right)\opt(v_1,\ldots,v_n) -O\left( \vsedit{m^2}\cdot \vsedit{n \cdot} \max_{S,i}v_i(S) \sqrt{\frac{1}{T}}\right).
\end{equation*}
If all bidders have XOS valuations with demand and XOS oracles the factor in front of OPT is $\frac{\lambda}{\max\{1,\mu_1\}+\mu_2}$.
\end{corollary}

We provide below two example applications of the latter theorems:
\paragraph{Application: Simultaneous Second Price Auctions.} Revisiting simultaneous second price auctions it is easy to see that the mechanism is $1$-value covered and $(0,1)$-threshold covered when players actions are restricted to no-overbidding actions and valuations are XOS. The value covering follows from the fact that for any set $S$, if we use as action $a_i^*(S)$, the bid vector that corresponds to the additive valuation returned by the XOS oracle for set $S$ (see proof of Lemma \ref{lem:from sets to bids}). As is shown in the proof of Lemma \ref{lem:from sets to bids}, this action satisfies that for any opponents action vector:
\begin{equation}
u_i(a_i^*(S),a_{-i})+\sum_{j\in S}\theta_{ij}(a_{-i}^t)\geq v(S),
\end{equation}
which is exactly the $1$-value covering inequality. The $(0,1)$-threshold covering inequality follows from the fact that for any feasible allocation $x=(S_1,\ldots,S_n)$, since players do not overbid:
\begin{equation}
\sum_{i\in [n]} \tau_i(S_i,a_{-i})\leq \sum_{j\in [m]} \max_{i\in [n]}b_{ij}\leq \sum_{i\in [n]} \sum_{j\in X_i(a)} b_{ij}\leq  \sum_{i\in [n]} v_i(X_i(a))\leq SW(a)
\end{equation} 
Thus we can apply the general theorems of this section with $\lambda=1$, $\mu_1=0$ and $\mu_2=1$ to recover the main theorems that we derived for SiSPAs in the previous sections. 

\paragraph{Application: Simultaneous First Price Auctions.} In a simultaneous first price auction at each item the bidder pays his own bid conditional on winning, rather than the second highest bid. Based on the proof of \cite{Syrgkanis2013}, that the simultaneous first price auction is $(1-1/e,1)$-smooth, it is easy to see that the mechanisms is actually $(1-1/e)$-value covered and $(1,0)$-threshold covered. Thus we can apply the latter theorems with $\lambda=1-1/e$, $\mu_1=1$ and $\mu_2=0$. 

\paragraph{Application: Simultaneous All-Pay Auctions.} In a simultaneous all-pay auction at each item the bidder pays his bid no matter whether he wins or not. Based on the proof of \cite{Syrgkanis2013}, that the simultaneous first price auction is $(1/2,1)$-smooth, it is easy to see that the mechanisms is actually $1/2$-value covered and $(1,0)$-threshold covered. Thus we can apply the latter theorems with $\lambda=1/2$, $\mu_1=1$ and $\mu_2=0$. 

\paragraph{Beyond additive threshold functions and combinatorial auctions.} Finding efficient algorithms for the online buyer's problem with general threshold functions is an interesting open problem that we defer to future work. Such an extension will enable an application of the approach presented in this paper to other value and threshold covered mechanisms, such as the greedy combinatorial auctions of \cite{Lucier2010}. Moreover, extending the algorithms for the online buyer's problem to valuations defined on mechanism design settings that are more general than combinatorial auction settings, such as the lattice valuations defined in \cite{Syrgkanis2013}, will also enable the generalization of our approach to compositions of more general mechanisms, such as position auctions. Both of these generalization seem fruitful future directions. Our approach here shows that the problem of efficient learning algorithms that retain welfare guarantees reduces to finding efficient learning algorithms for the online buyer's problem.


\section{Further Related Work}
\label{sec:app-related}

Closer to our intractability results is the work of Cai and Papadimitriou~\cite{Cai2014}, who show intractability of computing Bayesian-Nash equilibrium, as well as certain notions of Bayesian no-regret learning, in SiSPAs. In the Bayesian model each player's valuation is not fixed, but drawn from some distribution independently. They show that both computing best responses and a Bayes-Nash equilibrium in such a setting are {\tt PP}-hard. They also show that Bayesian coarse correlated equilibria are {\tt NP}-hard, and hence a certain type of Bayesian no-regret learning (namely when bidders re-sample their type in every round) is intractable. There are two important differences of their hardness results compared to ours: 
\begin{itemize}

\item First, the hardness of best response in their setting is driven by the fact that the opponent bids implicitly define a distribution of exponential support. In contrast, our inapproximability of best response is shown for an explicitly given opponent bid distribution.

\item Theirs is a setting where Bayesian coarse correlated equilbria are already hard, implying in particular that no-regret learning (with resampling of types in every round) is intractable. In contrast, in our setting \cite{Christodoulou2008} has provided a centralized polynomial-time algorithm for computing a pure Nash equilibrium in complete information SiSPAs with submodular bidders.  
Moreover, for some special cases of combinatorial auctions with submodular bidders, \cite{Dobzinski2015}, show that computing an equilibrium with good welfare is as easy as the algorithmic problem, ignoring incentives. The centralized nature of the algorithms in these papers and the complete information assumption make this result inherently different from the setting that we want to analyze, which is the \emph{agnostic} setting where players don't know anything about the game and behave in a decentralized manner. In particular, in our setting, the intractability comes from the distributed nature of the computation and the incomplete, non-Bayesian, information that the bidders have.
\end{itemize}

%
%

There is a large body of work on price of anarchy in auctions, in the incomplete information Bayesian/non-Bayesian setting and under no-regret learning behavior. We cannot do justice to the vast literature but here are some example papers: \cite{Christodoulou2008, Bhawalkar2011, Hassidim2011, Feldman2013, Markakis2012,Markakis2013,Lucier2010}. The price of anarchy of no-regret learning outcomes was first analyzed by \cite{Blum2008} in the context of routing games and was generalized to many games in \cite{Roughgarden2009} and to many mechanisms in \cite{Syrgkanis2013}, via the notion of smoothness. There is a strong connection between the smoothness framework and no-envy dynamics. In particular, the no-envy guarantee directly implies the lower bounds on the bidder's utility, which needed for the smoothness proof to go through. This is the main reason why no-envy implies price of anarchy guarantees.

Another major stream of work in algorithmic mechanism design addresses the design of computationally efficient dominant strategy truthful mechanisms \cite{Dobzinski2011,Dobzinski2012,Dobzinski2013,Dughmi2011,Dughmi2015,Dobzinski2016}. For instance, \cite{Dobzinski2011} shows that with only value queries, no distribution over deterministc truthful mechanisms can achieve better than polynomial approximations for submodular bidders. With demand queries \cite{Dobzinski2013} shows that no truthful in expectation mechanism can achieve better than $1-1/2e$-approximation. For coverage valuations \cite{Dughmi2011} gives a $1-1/e$-approximation, truthful in expectation randomized mechanism. For submodular bidders with demand queries the best truthful mechanism was recently given by \cite{Dobzinski2016} achieving $O(\sqrt{\log(m)})$-approximation. In contrast, our result shows that for no-envy XOS bidders with demand oracles, simultaneous item auctions achieve constant factor approximations.

Moreover, several papers address only the algorithmic problem of welfare maximization in combinatorial auctions with complement-free valuations. For instance, \cite{Feige2006} provides a 2-approximation for combinatorial auctions with sub-additive bidders and a $(1-1/e)$-approximation for XOS bidders, with access to demand oracles, improving upon prior work of \cite{Dobzinski2006} which also required XOS oracles. Our work can also be viewed as providing a simple and distributed algorithm for welfare maximization with XOS bidders, with a $(1-1/e)$-approximation guarantee: simply run our no-envy algorithms in a simultaneous first price auction game and then pick the best solution after a sufficient number of iterations. 

There is a large body of work on online learning and online convex optimization to which we cannot possibly do justice. We refer the reader to two recent surveys \cite{Bubeck2012,Shalev-Shwartz2012}. There is also a large body of work on online linear optimization where the number of experts is exponentially large, but the utility is linear in some low dimensional space. This setting was initiated by \cite{Kalai2005} and spurred a long line of work. We refer the reader to the relevant section of \cite{Bubeck2012}. Our results on perturbed leader algorithms generalize these results beyond the linear setting and we have provided some example applications beyond SiSPAs in Sections~\ref{sec:finite-parameter} and~\ref{sec:app-security}.

Our work is also related to the recent work of \cite{Hazan2015} on the power of best-response oracles in online learning. This paper gives query complexity lower bounds for the general online learning problem. In contrast, our approach defines sufficient conditions (the stability) under which best-response oracles are sufficient for efficient learning and hence optimization is equivalent to online learning. Therefore, we provide a positive counterpart to these negative results.

\bibliographystyle{alpha}
\bibliography{poa_survey}
\newpage

\begin{appendix}

\section{Omitted Proofs from Section~\ref{sec:hardness}}\label{sec:app-hardness}

\subsection{Proof of Theorem~\ref{thm:exact-optimal}}

\begin{rtheorem}{Theorem}{\ref{thm:approx-optimal}}
The optimal bidding problem is $\NP$-hard to approximate to within an additive $\xi$ even when: the threshold vectors in the support of (the explicitly described distribution) $D$ take values in $\{1,H\}^m$, where $H=k^2\cdot m^2$, $v=2\cdot k\cdot m$ and $\xi=\frac{1}{2k}$.
\end{rtheorem}
\begin{proof}
We break the proof in two Lemmas. In the first we show $\NP$-hardness of the exact problem and then we show hardness of the additive approximation problem. 

\begin{lemma}[Hardness of Optimal Bidding]
The optimal bidding problem is $\NP$-hard even if the threshold vectors in the support of $D$ take values in $\{1,H\}^m$, where $H=k^2\cdot m^2$ and $v=2\cdot k\cdot m$.
\end{lemma}
\begin{proof}

Before we move to the reduction we introduce some notation that is useful for the special case when thresholds are in $\{1,H\}^m$. First each threshold vector $p\in \{1,H\}^m$ can be uniquely represented by a set $T$, which corresponds to the items on which the threshold is $1$. Hence, the bid distribution $D$ can be therefore described by a collection of sets $\T=\{T_1,\ldots,T_k\}$, such that each set $T_i$ arises with probability $1/k$. 

Moreover, observe that in the optimization problem we might as well only consider strategies where the player a bid vector in $\{0,2\}^m$. \vsdelete{Bidding any bid in $[0,1)$ is equivalent to bidding $0$. Bidding $1$ is strictly dominated by bidding $2$ if $q\in [0,1)$ is equivalent to bidding $2$ when bidding $2$. Bidding anything in $(2,H)$ is equivalent to bidding $2$.}
\vsedit{Bidding any bid in $[0,1]$ is equivalent to bidding $0$. Bidding anything in $(1,H)$ is equivalent to bidding $2$.} Moreover, bidding in $[H,\infty)$ is dominated by bidding $2$. The reason is that bidding $H$ increases your probability of winning only in the cases when the threshold is $H= k^2 \cdot m^2$. But in those cases your utility is negative since $v=2\cdot k\cdot m$. Thus it is always optimal to remove those winning cases.

Thus any bidding strategy is also uniquely characterized by a set $S$, which is the set of items on which the player bids $2$. If a bidder chooses a set $S$, then he loses all items in $S$ only if a set $T_i$ arises, such that $T_i\cap S=\emptyset$, since then all items in $S$ have a threshold of $H$. Thus the probability that he wins some item is equal to:
\begin{equation}
\Pr[\text{win with $S$}] = 1- \frac{|\{T_i\in \T: T_i\cap S=\emptyset\}|}{k}
\end{equation}
Moreover, he pays only for the items for which he bids $2$ and only when the threshold of the item is also $1$. Thus his expected payment when he chooses $S$ is $\frac{1}{k}\sum_{T_i\in \T} |T_i\cap S|$. 

Therefore, the optimal bidding problem boils down to finding the set $S$ that maximizes:
\begin{equation}
\max_{S} v\cdot \left(1-\frac{|\{T_i\in \T: T_i\cap S=\emptyset\}|}{k}\right) -\frac{1}{k}\sum_{T_i\in \T} |T_i\cap S|
\end{equation}
Equivalently it is the problem of minimizing the negative part (since $v=2\cdot k\cdot m$ is a constant), multiplied by $k$:
\begin{equation}
\min_{S} v\cdot |\{T_i\in \T: T_i\cap S=\emptyset\}| +\sum_{T_i\in \T} |T_i\cap S|
\end{equation}

Moreover, we will only consider the case where $\T$ does not contain the empty set of items $[m]$ (i.e. some item always has a threshold of $1$). In this case observe that by picking the whole set of items the first part of the above objective vanishes, and the second part is at most $ k \cdot m$. Thus the value of the optimal objective is at most $k\cdot m$. 

Observe that if the player picks a set $S$, such that $\exists T_i\in \T: T_i\cap S=\emptyset$, then the first term is at least $v=2\cdot k\cdot m$. Thus any such solution $S$ must be suboptimal. Thereby at any optimal solution the first term in the objective vanishes to zero. Hence, the optimization problem for $v=2\cdot k\cdot m$, is equivalent to the problem:
\begin{equation}
\min_{S: \nexists T_i\in \T: T_i\cap S=\emptyset} \sum_{T_i\in \T} |T_i\cap S|
\end{equation}

We will reduce the set cover problem on a regular hypergraph to the latter equivalent form of the optimal bidding problem, which we will refer to as the simplified bidding problem. We consider a set cover instance with $k$ elements $\{t_1,\ldots,t_k\}$ and $m$ sets $\{C_1,\ldots,C_m\}$. All sets $C_j$ are of equal size $r$. 

For each set $C_j$ in the set cover problem we create an item $c_j$ in the bidding problem. For each element $t_i$ in the set cover problem we create a set $T_i$ in the simplified bidding problem defined as follows:
\begin{equation}
T_i = \{ c_j \in [m]: t_i\in C_j\}
\end{equation}
Thus set $T_i$ contains all the items $c_j$, corresponding to sets $C_j$ that contain element $t_i$.

First observe that for the specific instance of the simplified bidding problem that we created we can simplify the objective, because all sets $C_j$ have the same size $r$:
\begin{equation}
 \sum_{T_i\in \T} |T_i\cap S| = \sum_{c_j\in S} |\{T_i: c_j\in T_i\}| = \sum_{c_j\in S} |\{t_i: t_i\in C_j\}|  
 = \sum_{c_j \in S}r = r|S|
\end{equation}
Thus the simplified optimal bidding problem for the instance we created boils down to finding the set $S$ of minimum cardinality, which satisfies the condition $\{\nexists T_i\in \T: T_i\cap S=\emptyset\}$. We will refer to any such feasible set of items $S$ as a winning item-set (since it guarantees that the player always wins some item independent of the threshold vector). 

We will now show that there is a one-to-one correspondence between winning item-sets in the simplified bidding problem instance and set covers in the original set cover instance. From this we will then conclude that finding the minimum cardinality winning item-set problem will imply finding the minimum cardinality set cover problem and would complete the reduction.

Consider a set cover $Q\subseteq \T$ of the set cover instance. We claim that the set of corresponding items $S=\{c_j: C_j\in Q\}$ is a winning item-set. Suppose that there exists a set $T_i\in \T$ such that $T_i\cap S=\emptyset$. Now consider the element $t_i$. Since $t_i$ was covered by $Q$, it means that there exists a set $C_j\in Q$ such that $t_i\in C_j$. By construction of $T_i$ we know that $T_i$ contains item $c_j$. Since $c_j\in S$, we get that the difference $S\cap T_i$ is non-empty, a contradiction.

Consider a winning item-set $S$ and let $Q$ denote the corresponding collection of sets in the set cover instance. Suppose that $Q$ is not a set cover. Thus there exists an element $t_i$ that is not covered. This means that all the collections $C_j\in Q$ do not contain element $t_i$. Since all such collections $c_j$ are not part of item set $T_i$, we have that $T_i\cap S=\emptyset$. Thus $S$ cannot be a winning item-set.
\end{proof}

We now show the stronger version of the hardness result. This will be useful when using the hardness of the optimal bidding problem to imply the impossibility of efficiently computable learning algorithms with polynomial regret rates.

\begin{lemma}[Hardness of Approximately Optimal Bidding]
The optimal bidding problem is $\NP$-hard to approximate to within an additive $\xi$ even when: the threshold vectors in the support of $D$ take values in $\{1,H\}^m$, where $H=k^2\cdot m^2$, $v=2\cdot k\cdot m$ and $\xi=\frac{1}{2k}$.
\end{lemma}
\begin{proof}
We show that getting the value of the instance of the optimal bidding created in the proof of Theorem~\ref{thm:exact-optimal} to within an additive error $\xi=\frac{1}{2k}$, will imply a $3$-approximation to the original $r$-regular set cover instance for any $r$. It is known that it is NP-hard to approximate the optimal $r$-regular set cover to within a $\log(r)-O(\log\log(r))$ factor \cite{Trevisan2001,Feige1998}.

Let $APX$ be a $\xi$-additive approximation to the reduced instance of the optimal bidding problem and  $OPT$ the optimal value of the optimal bidding problem, i.e.:
\begin{equation}
OPT \in \left[ APX-\xi, APX+\xi\right]
\end{equation}
Let $OPT_c$ be the optimal value to the set cover instance. We already argued in the proof of Theorem~\ref{thm:exact-optimal} that:
\begin{equation}
OPT = v - OPT_c \frac{r}{k}
\end{equation}
Thus we get:
\begin{equation}
(v-APX)-\xi\leq OPT_c\frac{r}{k} \leq (v-APX)+\xi
\end{equation}
Now we know that the bidder in the bidding problem must always pay at least $1$ to win. Thus $APX\leq v-1\Leftrightarrow v-APX\geq 1$. Hence, we get that if:
\begin{equation}
\xi \leq \frac{1}{2k}\leq \frac{v-APX}{2k}
\end{equation}
then:
\begin{equation}
\left(1-\frac{1}{2k}\right)(v-APX)\leq (v-APX)-\xi\leq OPT_c\frac{r}{k} \leq (v-APX)+\xi\leq \left(1+\frac{1}{2k}\right)(v-APX)
\end{equation}
Which subsequently implies:
\begin{equation}
\frac{2k-1}{2\cdot r}(v-APX) \leq OPT_c \leq \frac{2k+1}{2\cdot r} (v-APX)
\end{equation}
Thus the value: $Q=\frac{2k-1}{2\cdot r}(v-APX)$ satisfies that:
\begin{equation}
Q\leq OPT_c \leq \frac{2k+1}{2 k-1} Q \leq 3\cdot Q
\end{equation}
i.e. it is a $3$-approximation to the value of the set cover instance.
\end{proof}
\end{proof}

\subsection{Proof of Theorem~\ref{thm:hardness of unit-demand learning}} \label{sec:app:final proof of lower bound}

\begin{proofof}{Theorem~\ref{thm:hardness of unit-demand learning}}
We have already provided a high-level sketch of the approach together with a discussion of the challenges that arise in Section~\ref{sec:hardness}. So we proceed directly with the technical details of the proof. The proof is via a contradiction. Towards this, we will suppose that there exists a no-regret learning algorithm for the online bidding problem that we are considering.

Consider the instances of the \emph{optimal bidding problem} obtained in the proof of Theorem~\ref{thm:exact-optimal}. In these instances, we have $v=2\cdot k\cdot m$ and a distribution $D$ of threshold vectors with support $k$. The threshold vectors in the support of the distribution take values in $\{1,k^2m^2\}^m$.

Now imagine that our no-regret learning algorithm faces threshold vectors that are guaranteed to lie in $\{1,k^2m^2\}^m$. In such a setting, every bid vector that our algorithm may submit is utility equivalent to or strictly dominated by some bid vector in $\{0,v/2m\}^m=\{0,k\}^m$. Thus we can assume that our no-regret learning algorithm only submits bid vectors from this set. (This is because, given a learning algorithm that does not satisfy this property, we can easily correct it to a better one that satisfies this property.) Hence, given the history of encountered threshold vectors until time step $t$, our no-regret learning algorithm will compute and submit a (potentially random) bid vector $b^t\in \{0,k\}^m$. Moreover, the no-regret learning property guarantees that, facing any sequence of threshold vectors $\{\theta^t\}_{t=1}^T$ from the set $\{1,k^2m^2\}^m$, the bid vectors submitted by our algorithm will satisfy:
\begin{equation}
\frac{1}{T} \sum_{t=1}^T \E\left[u(b^t,\theta^t)\right] \geq \max_{b\in \{0,k\}^m} \frac{1}{T} \sum_{t=1}^T u(b,\theta^t) - \epsilon(T)
\end{equation}
where $\epsilon(T) = poly(T^{-1},m,v)=poly(T^{-1},m,k)$. We will show how to use such no-regret learning guarantee to compute an approximate solution to an arbitrary instance of the optimal bidding problem from Theorem~\ref{thm:exact-optimal}.

We will achieve this as follows: Get $T$ i.i.d. samples of threshold vectors from distribution $D$. Denote this sequence of samples $\{\theta^t\}$. Run the no-regret learning algorithm against this sequence $N$ independent times. 

We denote by $b_z^t$ the bid vector submitted by the algorithm at time-step $t$ of the $z$-th execution. Let $V_z^T(\{\theta^t\})$ be the average realized utility of the algorithm in the $z$-th execution, i.e.
\begin{equation}
V_z^T(\{\theta^t\}) = \frac{1}{T} \sum_{t=1}^T  u(b_z^t,\theta^t)
\end{equation}
We will show that:
\begin{equation}
\left| \frac{1}{N} \sum_{z=1}^N V_z^T(\{\theta^t\})-\max_{b} \E_{p\sim D}\left[u(b,p)\right]\right|\leq c(T,\delta,N,k,m)
\end{equation}
with probability at least $1-\delta$ for some appropriately defined function $c(T,\delta,N,k,m)$.

\paragraph{Upper bound.} Observe that $V_z^T$ is drawn i.i.d. across the $z$ runs and its expectation is equal to:
\begin{equation}
V^T(\{\theta^t\}) = \E_{\{b^t\}}\left[\frac{1}{T} \sum_{t=1}^T u(b^t,\theta^t)\right]
\end{equation}
By Chernoff-Hoeffding bounds and since the average utility is upper bounded by $v$:
\begin{equation}
\Pr\left[\left|\frac{1}{N} \sum_{z=1}^N V_z^T(\{\theta^t\})- V^T(\{\theta^t\})\right| \geq \frac{q}{N}\right]\leq 2e^{-\frac{2q^2}{N\cdot v^2}}
\end{equation}
Since $v=2km$, if we set $c(N,\delta) = k\cdot m\cdot \sqrt{\frac{2\log(2/\delta)}{N}}$, then we know that with probability $1-\delta$:
\begin{equation}
\left|\frac{1}{N} \sum_{z=1}^N V_z^T(\{\theta^t\})- V^T(\{\theta^t\})\right|  \leq c(N,\delta)
\end{equation}

Let $b^*=\arg\max_{b\in \{0,k\}^m} \E_{p\sim D}\left[u(b,p)\right]$. By the no-regret assumption we know that:
\begin{equation}
V^T(\{\theta^t\}) \geq \max_{b} \frac{1}{T} \sum_{t=1}^T u(b,\theta^t)-\epsilon(T)\geq \frac{1}{T} \sum_{t=1}^T u(b^*,\theta^t)-\epsilon(T)
\end{equation}
Since $\theta^t$ are drawn i.i.d. from $D$, we know by Chernoff-Hoeffding bounds that:
\begin{equation}
\Pr\left[\left|\frac{1}{T} \sum_{t=1}^T u(b^*,\theta^t) - \E_{p\sim D}\left[u(b^*,p)\right]\right| \geq \frac{q}{T}\right]\leq 2e^{-\frac{2q^2}{T\cdot v^2}}
\end{equation}
Thus with probability $1-\delta$:
\begin{equation}
\left|\frac{1}{T} \sum_{t=1}^T u(b^*,\theta^t) - \E_{p\sim D}\left[u(b^*,p)\right]\right| \leq c(T,\delta)
\end{equation}

Combining the above we get that with probability $1-2\delta$:
\begin{align*}
\frac{1}{N} \sum_{z=1}^N V_z^T(\{\theta^t\})\geq~& V^T(\{\theta^t\})-c(T,\delta)\\
 \geq~&  \frac{1}{T} \sum_{t=1}^T u(b^*,\theta^t)-\epsilon(T) - c(T,\delta)\\
\geq~&  \E_{p\sim D}\left[u(b^*,\theta^t)\right]-\epsilon(T) - c(T,\delta)-c(N,\delta)
\end{align*}

\paragraph{Lower bound.} Consider the random variable $X_z^t = u(b_z^t,\theta^t)- \E_{p\sim D}\left[u(b_z^t,p)\right]$. Let $\F_{t-1}$ denote the filtration of all the information observed by the algorithm up till time-step $t-1$ in the $z$-th execution. This is the thresholds and bids in the past steps. Observe that the conditional expectation (over bids and thresholds) of the latter variable is:
\begin{equation}
\E[X_z^t ~|~ \F_{t-1}] = \E\left[u(b_z^t,\theta^t)~|~\F_{t-1}\right] - \E\left[\E_{p\sim D}\left[u(b_z^t,p)\right]~|~\F_{t-1}\right]
\end{equation}
Since $\theta^t$ are drawn i.i.d. at each time-step and are not observed by the algorithm before deciding $b^t$, we have that:
\begin{equation}
\E\left[u(b_z^t,\theta^t)~|~\F_{t-1}\right]=\E\left[\E_{p\sim D}\left[u(b_z^t,p)\right]~|~\F_{t-1}\right]
\end{equation}
Thus we have that $\E\left[X_z^t~|~\F_{t-1}\right]=0$ and thereby $\{X_z^t\}_t$ is a bounded martingale difference sequence, with $|X_z^t|\leq v$.  Hence, by Hoeffding-Azuma inequality:
\begin{equation}
\Pr\left[\left|\sum_{t=1}^T X_z^t \right| \geq q\right] \leq 2e^{-\frac{2q^2}{T\cdot v^2}}
\end{equation}
The latter implies that with probability $1-\delta$:
\begin{equation}
\left|\frac{1}{T} \sum_{t=1}^T u(b_z^t,\theta^t) - \frac{1}{T} \sum_{t=1}^T \E_{p\sim D}\left[u(b_z^t,p)\right]\right|\leq c(T,\delta)
\end{equation}
Hence, we also get that for each $z$ with probability $1-\delta$:
\begin{align*}
V_z^{T}(\{\theta^t\})=\frac{1}{T} \sum_{t=1}^T u(b_z^t,\theta^t) \leq~&  \frac{1}{T} \sum_{t=1}^T \E_{p\sim D}\left[u(b_z^t,p)\right] +c(T,\delta)\\
\leq~& \frac{1}{T} \sum_{t=1}^T \max_{b}\E_{p\sim D}\left[u(b,p)\right] +c(T,\delta)\\
=~& \max_{b}\E_{p\sim D}\left[u(b,p)\right] +c(T,\delta)
\end{align*}
The latter holds for all $z$ with probability at least $1-N\cdot \delta$. Therefore, with probability at least $1-N\delta$, the average across the $N$ runs will satisfy the above bound. 

\paragraph{Concluding.} Hence we can conclude that with probability at least $1-(N+2)\delta$:
\begin{equation}
-\epsilon(T)-c(N,\delta)-c(T,\delta)\leq\left( \frac{1}{N} \sum_{z=1}^N V_z^T(\{\theta^t\})-\max_{b} \E_{p\sim D}\left[u(b,p)\right]\right)\leq c(T,\delta)
\end{equation}

Picking $\delta = \frac{\zeta}{N+2}$ and $N=T$, we get that with probability $1-\zeta$ we have:
\begin{equation}
\left|\frac{1}{N} \sum_{z=1}^N V_z^T(\{\theta^t\})-\max_{b} \E_{p\sim D}\left[u(b,p)\right]\right|\leq \epsilon(T) + 2\cdot k\cdot m\cdot \sqrt{\frac{2\log(2 (T+2)/\zeta)}{T}}
\end{equation}

Thus by doing $T$ executions against a random sequence of threshold vectors from $D$ of length $T$ and averaging the average utilities gives us an approximation to the value of the offline \emph{optimal bidding} problem, $\max_{b} \E_{p\sim D}\left[u(b,p)\right]$. Specifically, if $\epsilon(T)=\poly(T^{-1},k,m)$ then by setting $T=\poly(\frac{1}{\xi},k,m)$ we can get the value of the latter problem to within an additive error of $\xi$. Thus we cannot possibly have $\epsilon(T) = \poly(T^{-1},k,m)$, since then $\poly(k,m)$  accesses to the polynomial time no-regret algorithm would give us a $\frac{1}{2k}$ additive approximation to the optimal bidding problem in polynomial time and with high probability.
\end{proofof}

\subsection{Intepretation of Theorem~\ref{thm:hardness of unit-demand learning}} \label{sec:sketch of lower bound}

Theorem~\ref{thm:hardness of unit-demand learning} can be viewed as a corrollary of two results, of which one is specific to SiSPAs and the other is a general claim about online learning.

\begin{itemize} 
\item  Theorem~\ref{thm:approx-optimal}  is equivalent to saying that it is {\tt NP}-hard to compute one step of the Follow-The-Leader (FTL) algorithm in SiSPAs, even for a unit-demand bidder with the same value for all items, and even when this value is given in unary representation.  Every step of FTL needs to compute an optimal bid vector against the empirical distribution of threshold bids that the algorithm has already encountered. Theorem~\ref{thm:approx-optimal} implies that this is {\tt NP}-hard.

\begin{theorem}[Corollary of Theorem~\ref{thm:approx-optimal}] \label{thm:one step of FTL is hard} Computing one step of FTL in SiSPAs is  {\tt NP}-hard, even for a unit-demand bidder with the same value $v$ for all items that is given in unary. In fact, 
it is {\tt NP}-hard to even compute the expected utility resulting from the optimum bid against the distribution of past opponent bids to within an additive {$m/v$}.
\end{theorem}

\item Even though FTL is not itself a no-regret learning algorithm, and even though we only established that one step of it is intractable, this sufficed to actually show  (see Section~\ref{sec:app:final proof of lower bound}) that there are no polynomial-time implementable no-regret learning algorithms in SiSPAs. Here, we comment that this is a general phenomenon, applicable to any online learning setting: namely, in any setting where one step of FTL is inapproximable, there is no polynomial-time no-regret learning algorithm either. This is summarized in the following theorem. 
\begin{theorem} \label{thm:LB for FTL implies general LB}
Consider a family  ${\cal F}$ of functions $f:{\cal X} \rightarrow [0,v]$, and suppose that $b$ bits suffice to index each function in $\cal F$ and element in $\cal X$ under some encoding. Suppose also that, given an explicit description of a distribution ${\cal D}$ over ${\cal F}$\footnote{The size of the description is the total number of bits needed to index the functions in its support and the probabilities assigned to them.} as well as $v$ in unary representation, it is {\tt NP}-hard to find some $x \in {\cal X}$ whose expected value $\E_{f \sim D}[f(x)]$ is within an additive $O(1/|D|^c)$ of the optimum for some constant $c>0$, where $|D|$ is the size of $D$'s support. Then, unless {\tt RP} $\supseteq$ {\tt NP}, there is no learning algorithm running in time polynomial in $b$, $v$, and $T$ and whose regret after $T$ steps is any polynomial in $b$, $v$, and $1/T$.
\end{theorem}
Observe that the hypothesis of Theorem~\ref{thm:LB for FTL implies general LB} is tantamount to the problem facing FTL in a learning environment with cost functions ${\cal F}$, which need not be linear of convex. Our statement is also careful to cater to pseudo-polynomial dependence of the running time and regret bound on the diameter $v$ of the range of our cost functions, as the typical dependence of no-regret learning algorithms on the diameter is pseudo-polynomial, and we are seeking to get lower bounds for such learners.

\begin{proofskof}{Theorem~\ref{thm:LB for FTL implies general LB}}
We simply observe that in the proof of Theorem~\ref{thm:hardness of unit-demand learning} in Section~\ref{sec:app:final proof of lower bound} we did not use the structure of the actual learning problem other than the fact that it is $\NP$-hard to find an approximately optimal response to an explicitly given distribution, which guarantees an additive error that is inverse-polynomial in the input.
\end{proofskof}

\end{itemize}

\section{Omitted Proofs from Section~\ref{sec:no-envy}}
\label{sec:app-envy}
\begin{proofof}{Lemma~\ref{lem:from sets to bids}}
Let $\A$ be the algorithm for the online \buyers problem. At every iteration $t$, we will query algorithm $\A$ for a set $S^t$. We will then call the XOS oracle on $S^t$ to get an additive valuation $\ell \in \Ell$ such that $v(S^t)=\sum_{j\in S^t} a_j^\ell$. We will submit a bid $b_j^t=a_j^t\cdot 1\{j\in S^t\}$ on each item $j$. Then after seeing the threshold vector $\theta^t$, we feed it to algorithm $\A$ as a threshold vector for the \buyers problem. We now show that this algorithm is an $\alpha$-approximate no-envy algorithm for the \emph{online bidding problem}.

First we argue that, based on the latter construction for any realization of the price vector $\theta^t$:\footnote{It is interesting to remark that the latter property has strong connection to the fact that the simultaneous second price auction is a smooth mechanism as defined in \cite{Syrgkanis2013} and the proof is similar to the proof of showing that the auction is smooth.}
\begin{equation*}
u(b^t,\theta^t) \geq u(S^t,\theta^t) = v(S^t)-\sum_{j\in S^t}\theta_j^t
\end{equation*}
Consider an arbitrary threshold vector $\theta^t$ and let $X^t=\{j: b_j^t>\theta_j^t\}$ be the subset of items of $S^t$ that the player won. By the definition of the XOS valuation we have that, $v(X^t) \geq \sum_{j\in X^t} a_j^\ell$. Thus:
\begin{align*}
u(b^t,\theta^t) =~& v(X^t) - \sum_{j\in X^t} \theta_j^t \geq \sum_{j\in X^t}a_j^\ell - \sum_{j\in X^t} \theta_j^t = \sum_{j\in S^t} (a_j^\ell - \theta_j^t)\cdot 1\{a_j^t>\theta_j^t\}
\geq \sum_{j\in S^t} (a_j^\ell - \theta_j^t)\\
 =~& v(S^t)-\sum_{j\in S^t}\theta_j^t
\end{align*}

Thus we have that the expected reward of the online bidding algorithm we constructed satisfies that for any adaptively chosen sequence of thresholds:
\begin{equation}
\E\left[ \frac{1}{T} \sum_{t=1}^T u(b^t,\theta^t)\right] \geq \E\left[ \frac{1}{T} \sum_{t=1}^T u(S^t,\theta^t)\right],
\end{equation} 
where $S^t$ is the random set that the \emph{online \buyers algorithm} would have chosen at time-step $t$, for the same sequence of adaptively chosen thresholds. By the guarantee of the \emph{online \buyers algorithm} $\A$, from Equation \eqref{eqn:alg-buyers-guarantee}, we can then conclude that the algorithm we created for the bidding problem is $\alpha$-approximate no-envy.

It remains to argue that our algorithm never submits an overbidding bid. This follows from the definition of XOS valuations. If $\ell$ is the index of the additive valuation that we picked at a time-step $t$, then for any set $X$, we have: $v(X) \geq \sum_{j\in X} a_j^\ell\geq \sum_{j\in X} b_j^t$.
\end{proofof}

\begin{proofof}{Theorem~\ref{thm:envy-poa}}
Let $X^*=(X_1^*,\ldots,X_n^*)$ be the optimal allocation of items to bidders. Moreover, for each item $j$ let $B_j=\max_{i\in [n]} b_{ij}$. By the no-envy property of each player we get:
\begin{align*}
\E\left[\frac{1}{T}\sum_{i=1}^T u_i(b^t)\right] \geq \frac{1}{\alpha}v_i(X_i^*) - \E\left[\sum_{j\in X_i^*} \hat{\theta}_{ij}\right] - \epsilon(T)
\end{align*}
Summing over all players we get:
\begin{align*}
\E\left[\frac{1}{T}\sum_{i=1}^T \sum_{i\in [n]}u_i(b^t)\right] \geq~& \frac{1}{\alpha}\sum_{i\in [n]}v_i(X_i^*) - \E\left[\sum_{i\in[n]} \sum_{j\in X_i^*} \hat{\theta}_{ij}\right] - n\cdot \epsilon(T)\\
\geq~& \frac{1}{\alpha}\opt - \E\left[\frac{1}{T}\sum_{t=1}^T\sum_{i\in[n]} \sum_{j\in X_i^*} B_j^t\right]-n\cdot \epsilon(T)\\
\geq~& \frac{1}{\alpha}\opt - \E\left[\frac{1}{T}\sum_{t=1}^T\sum_{j\in [m]} B_j^t\right]-n\cdot \epsilon(T)
\end{align*}
Let $X_i^t$ denote the random set that player $i$ acquired at time-step $t$. Since the utility of the player is at most his value, we have: $v_i(X_i^t)\geq u_i(b^t)$. Moreover, since players do not overbid, we have that for any realization of the randomness: $\sum_{i\in[n]} v_i(X_i^t) \geq \sum_{j\in [m]}B_j^t$. Hence:
\begin{align*}
\E\left[\frac{1}{T}\sum_{i=1}^T \sum_{i\in [n]}v_i(X_i^t)\right]
\geq~& \frac{1}{\alpha}\opt - \E\left[\frac{1}{T}\sum_{t=1}^T\sum_{i\in [n]} v_i(X_i^t)\right]-n\cdot \epsilon(T)
\end{align*}
By re-arranging we get the theorem.
\end{proofof}

\section{Omitted Proofs from Section~\ref{sec:oracles}}

\subsection{From adaptive to oblivious adversaries}\label{sec:app-oblivious} 
We will utilize a generic reduction provided in Lemma 12 of \cite{Hutter2005}, which states that given that in Algorithm~\ref{defn:ftpl} we draw independent randomization at each iteration, it suffices to provide a regret bound only for oblivious adversaries, i.e., the adversary picks a fixed sequence $\theta^{1:T}$ ahead of time without observing the actions of the player. Moreover, for any such fixed sequence of an oblivious adversary, the expected utility of the algorithm can be easily shown to be equal to the expected utility if we draw a single random sequence $\{x\}$ ahead of time and use the same random vector all the time. 

The proof is as follows: by linearity of expectation and the fact that each sequence $\{x\}^t$ drawn at each time-step $t$ is identically distributed:
\begin{align*}
\E_{\{x\}^1,\ldots,\{x\}^t}\left[\sum_{t=1}^T u(M(\{x\}^t\cup \theta^{1:t-1}),\theta^t)\right] =~& \sum_{t=1}^T \E_{\{x\}^t}\left[u(M(\{x\}^t\cup \theta^{1:t-1}),\theta^t)\right]\\
 =~& \sum_{t=1}^T \E_{\{x\}^1}\left[u(M(\{x\}^1\cup \theta^{1:t-1}),\theta^t)\right]\\
=&\E_{\{x\}^1}\left[\sum_{t=1}^T u(M(\{x\}^1\cup \theta^{1:t-1}),\theta^t)\right]
\end{align*}
The latter is equivalent to the expected reward if we draw a single random sequence $\{x\}$ ahead of time and use the same random vector all the time. Thus it is sufficient to upper bound the regret of this modified algorithm, which draws randomness only once.


\subsection{Proof of Theorem~\ref{thm:general-regret}}

\begin{proofof}{Theorem~\ref{thm:general-regret}}
To prove the theorem it suffices to show that the regret of the algorithm which gets to observe $\theta^t$ ahead of time and at each time-step plays action $\tilde{a}^t=\Or{\{x\}\cup\theta^{1:t}}$ is upper bounded by the second term in the bound. Then the theorem easily follows by observing that the expected reward of Algorithm~\ref{defn:lazy-ftpl} is close to the reward of this fore-sight algorithm by an error which is equal to the first term in the bound.  

Thus in the remainder of the section we will analyze this algorithm with fore-sight, and bound its regret in a sequence of two Lemmas. 
\begin{lemma}[Be-the-leader lemma]\label{lem:btl} Suppose that we actually learned $\theta^t$ ahead of time-step $t$ and we played according to action $\tilde{a}^t=\Or{\theta^{1:t}}$. Then for any sequence $\theta^{1:T}$, this algorithm has no-regret against any action $a^*$, i.e.:
\begin{equation}
\sum_{t=1}^T u(\Or{\theta^{1:t}},\theta^t) \geq \sum_{t=1}^T u(\Or{\theta^{1:T}},\theta^t)
\end{equation}
\end{lemma}
\begin{proof}
We show it by induction. The induction hypothesis is that for any $t\in [1:T]$: 
\begin{equation*}
\sum_{\tau=1}^{t} u(\Or{\theta^{1:\tau}},\theta^\tau)\geq U(\Or{\theta^{1:t}},\theta^{1:t})
\end{equation*}
This trivially holds for $t=1$. Assume it holds for $t$ and consider the case of $t+1$:
\begin{align*}
\sum_{\tau=1}^{t+1} u(\Or{\theta^{1:\tau}},\theta^\tau)=~&\sum_{\tau=1}^{t} u(\Or{\theta^{1:\tau}},\theta^\tau)+ u(\Or{\theta^{1:t+1}},\theta^{t+1})\\
\geq~& U\left(\Or{\theta^{1:t}},\theta^{1:t}\right)+  u(\Or{\theta^{1:t+1}},\theta^{t+1})\\
\geq~& U(\Or{\theta^{1:t+1}},\theta^{1:t})+ u(\Or{\theta^{1:t+1}},\theta^{t+1})\\
=~& U(\Or{\theta^{1:t+1}},\theta^{1:t+1})
\end{align*}
which concludes the induction. 
\end{proof}

\begin{lemma}[Be-the-leader with fixed sample perturbations]\label{lem:pbtl} Suppose that we actually learned $\theta^t$ ahead of time-step $t$ and we played according to action $\tilde{a}^t=M(\{x\}\cup \theta^{1:{t}})$, i.e. including $\theta^t$ in the frequency vector for some fixed sequence $\{x\}=\{x^1,\ldots,x^k\}$ of $k$ parameters. Then, this algorithm achieves regret against any action $a^*$:
\begin{equation}
\sum_{t=1}^T \left( u(a^*,\theta^t) - u(\tilde{a}^t,\theta^t)\right) \leq  \sum_{\tau=1}^k \left(f_+(x^\tau)+f_-(x^\tau)\right)
\end{equation}
\end{lemma}
\begin{proof}
Denote with $k$ the length of sequence $\{x\}$. Consider the sequence $\{x\}\cup \theta^{1:T}$ and apply the \emph{be-the-leader} Lemma~\ref{lem:btl} to this sequence. We get that for any action $a^*$:
\begin{align*}
\sum_{\tau=1}^{k} u(\Or{x^{1:\tau}},x^\tau) + \sum_{t=1}^T u(\Or{\{x\}\cup\theta^{1:t}},\theta^t) \geq 
\sum_{\tau=1}^{k} u(a^*,x^\tau) + \sum_{t=1}^T u(a^*,\theta^t)
\end{align*}
Since for all $a\in A: -f_-(\theta)\leq u(a,\theta) \leq f_+(\theta)$ for any $\theta$, and since $\tilde{a}^t = M(\{x\}\cup\theta^{1:t})$, by re-arranging the latter inequality we get:
\begin{equation*}
\sum_{t=1}^T \left( u(a^*,\theta^t) - u(\tilde{a}^t,\theta^t)\right) \leq \sum_{\tau=1}^k \left( u(\Or{x^{1:\tau}},x^\tau)-u(a^*,x^\tau)\right) \leq  \sum_{\tau=1}^k \left(f_+(x^\tau)+f_-(x^\tau)\right)
\end{equation*}
\end{proof}

\end{proofof}

\section{Omitted Proofs from Section~\ref{sec:demand}}
\subsection{Proof of Theorem~\ref{thm:demand-regret-bound}}\label{sec:app-demand-thm}

\textbf{Stability bound.} The most important part of the analysis is providing the stability bound $g(t)$, for each time-step $t$, given the perturbation we used. We will show it in the following lemma. 
\begin{lemma}[Stability lemma]\label{lem:demand-stability} For any sequence $\theta^{1:T}$ such that $0\leq\theta_j^t\leq D$ for all $j\in[m]$, $t\in [T]$, the stability of Algorithm~\ref{defn:lazy-ftpl} when applied to the online \buyers problem and with the single-sample exponential perturbation is upper bounded by:
\begin{equation}
\E_{x}\left[u(M(x\cup \theta^{1:t}),\theta^t)-u( M(x\cup \theta^{1:{t-1}}), \theta^t)\right]\leq  \vsedit{(mD+H)} \left(\frac{m}{t} +3\epsilon mD\right)
\end{equation}
\end{lemma}
\begin{proof}
We consider a specific sequence $\theta^{1:T}$ and a specific time-step $t$. \vsedit{First observe that it suffices to prove the lemma for the translated utilities $u'(a,\theta)= u(a,\theta)+mD$, since such a translation preserves utility differences for any two actions and does not alter the optimizing actions in hindsight. These translated utilities lie in $[0,mD+H]$, when $\theta_j\leq D$. Thus we will prove the lemma for utilities $u(a,\theta)$ that lie in $[0,mD+H]$.} For succinctness we we will denote with: 
\begin{align*}
\FTPL^t =~& \E_{x}\left[u( M(x\cup \theta^{1:{t-1}}), \theta^t)\right]\\
\BTPL^t =~& \E_{x}\left[u(M(x\cup \theta^{1:t}),\theta^t)\right]
\end{align*}
the expected reward of \emph{follow the perturbed leader} and \emph{be the perturbed leader}, correspondingly, at time-step $t$. 

We will construct a mapping $\mu:\Theta\rightarrow \Theta$, such that for any random sample $x$, we will have that:
\begin{equation}
M(x\cup \theta^{1:t}) = M(\mu(x)\cup \theta^{1:t-1})
\end{equation}
Observe that the two maximization problems are the same if we have that the average threshold vector ends up being the same, i.e.,
\begin{equation}
\frac{\sum_{\tau=1}^t \theta^\tau+x}{t+1} = \frac{\sum_{\tau=1}^{t-1}\theta^\tau+\mu(x)}{t}
\end{equation}
By re-arranging we get:
\begin{equation}
\mu(x) = \frac{t}{t+1} x - \frac{1}{t+1} \sum_{\tau=1}^{t-1}\theta^\tau + \frac{t}{t+1}\theta^t
\end{equation}
Observe that $\mu(x)$ is a bijection from $\Theta^+\triangleq\{x\in \Theta:\mu(x)\geq 0\}$ to $\Theta$ (where by $\mu(x)\geq 0$ we mean a coordinate-wise comparison). Thus we can write:
\begin{align*}
\FTPL^t=~& \int_{x\in \Theta}u(M(x\cup \theta^{1:t-1}),\theta^t)f(x) dx\\
=~& \int_{y\in \Theta^+}u(M(\mu(y)\cup \theta^{1:t-1}),\theta^t) f(\mu(y)) |det(\nabla\mu)(y)| dy\\
=~&\int_{y\in \Theta^+}u(M(y\cup \theta^{1:t}),\theta^t) f(\mu(y)) \left(\frac{t}{t+1}\right)^m dy\\
=~&\int_{x\in \Theta^+}u(M(x\cup \theta^{1:t}),\theta^t) f(\mu(x)) \left(\frac{t}{t+1}\right)^m dx
\end{align*}
Now observe that for any $x\in \Theta^+$:
\begin{align*}
f(\mu(x)) =~& \exp\left\{-\epsilon\left(\|\mu(x)\|_1-\|x\|_1\right)\right\}f(x)\\
=~& \exp\left\{-\epsilon\left(\left\|\frac{t}{t+1} x - \frac{1}{t+1} \sum_{\tau=1}^{t-1}\theta^\tau + \frac{t}{t+1}\theta^t\right\|_1 - \|x\|_1\right)\right\} f(x)\\
\geq~&\exp\left\{-\epsilon\left(\frac{t}{t+1} \|x\|_1 + \frac{1}{t+1} \left\|\sum_{\tau=1}^{t-1}\theta^\tau\right\|_1 + \frac{t}{t+1}\|\theta^t\|_1 - \|x\|_1\right)\right\} f(x)\\
\geq~&\exp\left\{-\epsilon 2 mD\right\} f(x)\\
\geq~&\left(1-2\epsilon mD\right) f(x)
\end{align*}
Moreover, it is easy to check that $\left(\frac{t}{t+1}\right)^m \geq 1-\frac{m}{t}$ and that for any two non-negative numbers $x,y$, $(1-x)(1-y)\geq 1-x-y$. Last, we remind \vsedit{that $0\leq u(a,\theta^t)\leq mD+H$} for any $a\in A$. Plugging the above lower bounds in the integral we get:
\begin{align}
\FTPL^t\geq~& \int_{x\in \Theta^+}u(M(x\cup \theta^{1:t}),\theta^t) \left(1-2\epsilon mD\right)\left(1-\frac{m}{t}\right) f(x)  dx\nonumber\\
\geq~& \int_{x\in \Theta^+}u(M(x\cup \theta^{1:t}),\theta^t) \left(1-2\epsilon mD-\frac{m}{t}\right) f(x)  dx\nonumber\\
\geq~& \int_{x\in \Theta}u(M(x\cup \theta^{1:t}),\theta^t) \left(1-2\epsilon mD-\frac{m}{t}\right)  f(x)  dx\nonumber - \vsedit{(mD+H)}\cdot\Pr[x\notin \Theta^+]\nonumber\\
\geq~& \int_{x\in \Theta}u(M(x\cup \theta^{1:t}),\theta^t) f(x) dx- \vsedit{(mD+H)}\left(2\epsilon mD+\frac{m}{t}+ \Pr[x\notin \Theta^+]\right)\nonumber\\
=~&\BTPL^t - \vsedit{(mD+H)}\left(2\epsilon mD+\frac{m}{t}+\Pr[x\notin \Theta^+]\right)\label{eqn:integral-lower-bound}
\end{align}
Now, let us examine the quantity $\Pr[x\notin\Theta^+]=1-\Pr[\mu(x)\geq 0]$. The inequality $\mu(x)\geq 0$ is equivalent to:
\begin{equation*}
t x - \sum_{\tau=1}^{t-1}\theta^\tau + t\theta^t \geq 0 \Leftrightarrow x\geq \frac{1}{t}\sum_{\tau=1}^{t-1}\theta^\tau - \theta^t
\end{equation*}
Observe that: $D\geq \frac{1}{t}\sum_{\tau=1}^{t-1}\theta^\tau - \theta^t$, thereby, the condition $\{x\geq D\}$ coordinate-wise, implies the condition $\{\mu(x)\geq 0\}$. Thus:
\begin{align*}
\Pr[\mu(x)\geq 0]\geq \Pr[x\geq D] = \prod_{j=1}^m \Pr[x_j\geq D] = e^{-\epsilon D m}\geq 1-\epsilon D m.
\end{align*}
Thus we get:
\begin{align}\label{eqn:exponential-cdf}
\Pr[x\notin\Theta^+] = 1-\Pr[\mu(x) \geq 0]\leq \epsilon D m
\end{align}
Plugging the bound from Equation \eqref{eqn:exponential-cdf} into Equation \eqref{eqn:integral-lower-bound} we get:
\vsedit{\begin{align*}
\FTPL^t\geq~&\BTPL^t- (mD+H) \left(3\epsilon mD+\frac{m}{t}\right)
\end{align*}}
This concludes the proof of the stability property.
\end{proof}

Now we can apply Theorem~\ref{thm:general-regret} to get a concrete bound for the follow-the-perturbed leader algorithm with the single sample exponential perturbation.

\begin{proofof}{Theorem~\ref{thm:demand-regret-bound}}
By applying Theorem~\ref{thm:general-regret} and Lemma~\ref{lem:demand-stability}, we get that the regret is upper bounded by:
\begin{align*}
\sum_{t=1}^T g(t) +  \E_{x}\left[f_+(x)+f_-(x)\right]\leq~& \sum_{t=1}^T \vsedit{(mD+H)} \left(\frac{m}{t} +3\epsilon mD\right) + \E_{x}\left[H+\|x\|_1\right]\\
\leq~& \vsedit{(mD+H)} \sum_{t=1}^T\frac{m}{t} +3\vsedit{(mD+H)}\epsilon mDT + H+ \frac{m}{\epsilon}\\
\leq~& \vsedit{(mD+H)} m (\log(T) + 1) +3\vsedit{(mD+H)}\epsilon mDT + H+ \frac{m}{\epsilon}
\end{align*}
Picking $\epsilon=\sqrt{\frac{1}{\vsedit{(mD+H)} D T}}$, we get:
\begin{align*}
\sum_{t=1}^T g(t) +  \E_{x}\left[f_-(x)+f_+(x)\right]\leq~&
 \vsedit{(mD+H)} m (\log(T) + 1)+H +4m \sqrt{\vsedit{(mD+H)}DT}
\end{align*}
\end{proofof}

\section{Omitted Proofs from Section~\ref{sec:convex}}
\label{sec:app-convex}

\subsection{Proof of Lemma~\ref{lem:buyers-learning}}

\paragraph{From the Online \Buyers Problem to Online Convex Optimization.}
Suppose that the buyer picks a set at each iteration at random from a distribution where each item $j$ is included independently with probability $x_j$ to the set. Any such distribution is characterized by its vector of marginals $x\in [0,1]^m$. For succinctness we will be denoting with $\Po=[0,1]^m$. The expected utility of the buyer at time-step $t$, from picking a set $S^t$ from such a distribution with marginals $x^t$ can be written as:
\begin{equation}
\E_{S^t \sim x^t}\left[u(S^t,\theta^t)\right] = \E_{S^t\sim x^t}\left[v(S^t)\right] - \langle\theta^t, x^t\rangle,
\end{equation}
where we denote with $\langle x, y\rangle$ the inner product between vectors $x$ and $y$.
The function $V(x) \equiv \E_{S\sim x}\left[v(S)\right]$ is what is called the {\em multi-linear extension of set function $v$}. Thus we can write the expected reward of the buyer from distribution $x^t$ as:
\begin{equation}
\E_{S^t \sim x^t}\left[u(S^t,\theta^t)\right] = V(x^t) - \langle\theta^t, x^t\rangle.
\end{equation}
When $V: \Po\rightarrow \R_+$ is a concave function, we can solve the online problem that the buyer faces by invoking some online convex optimization algorithm such as the online gradient descent with projections of \cite{Zinkevich2003}. All we need in order to invoke such an algorithm in this case is that we can compute the gradient of $V(x)$. 
If both conditions were satisfied, it would be easy to get a polynomial-time learning algorithm for computing the distribution of sets $x^t$ at each time-step so as to achieve polynomial regret  rates for the online \buyers problem. 
%

However, the multi-linear extension is not a concave function for most natural classes of valuation functions. This same issue arises when designing truthful-in-expectation mechanisms for welfare maximization, such as in \cite{Dughmi2011}. We will exploit the approach taken in that paper to solve the problem we face here. Instead of set-distributions that are independent across items, the idea is to use correlated set-distributions that both convexify the expected utility of the buyer, and maintain the polynomial-time computability of its gradient. 

\paragraph{Convex Rounding.}
For any item-independent distribution $x^t$ we will construct a distribution $D(x^t)$ and then draw a set $S^t$ from $D(x^t)$. The distribution that we will construct will have the following properties:
\begin{defn}[Convex rounding scheme.] \label{def:convex rounding}A mapping $D:[0,1]^m\rightarrow \Delta(2^m)$ is an $\alpha$-approximate convex rounding scheme if:
\begin{enumerate}
\item The function $F(x) = \E_{S\sim D(x)}\left[v(S)\right]$ is a concave function of $x$.
\item The gradient of $F$ is polynomial-time computable. 
\item If we denote by $y_j(x)$ the probability that item $j$ is included in $S$ when $S$ is drawn from distribution $D(x)$, then  $y_j(x)\leq x_j$.
\item For any integral $x$, which corresponds to a set $S$: $F(x) \geq \frac{1}{\alpha} v(S)$.
\end{enumerate}
\end{defn}

If we can find such an \emph{$\alpha$-approximate convex rounding scheme}, then we show that we can construct an $\alpha$-approximate no-envy algorithm for the online \buyers problem.
\begin{lemma}\label{lem:online-convex}
If the valuation function $v(\cdot)$ admits an $\alpha$-approximate convex rounding scheme with $\sup_{x} \|\nabla F(x)\|_2\leq H$, then there exists a polynomial-time computable learning algorithm for the online \buyers problem which guarantees:
\begin{equation}
\E\left[\frac{1}{T}\sum_{t=1}^T u(S^t,\theta^t)\right]\geq \max_{S}\left(\frac{1}{\alpha}v(S) - \sum_{j\in S}\hat{\theta}_j^T\right) -3(H+\sqrt{mK})\sqrt{\frac{m}{T}}.\end{equation}
where $K=\max_{t\in [T],j\in [m]}\theta_j^t$.
\end{lemma}
\begin{proof}
Our online buyer algorithm will invoke the projected gradient descent algorithm of \cite{Zinkevich2003} (see also \cite{Hazan2006} for a slightly improved analysis) applied to the following convex optimization problem:
At each iteration the algorithm picks a vector $x^t\in [0,1]^m$ and receives a reward of:
\begin{equation}
f^t(x^t) = F(x^t) - \langle \theta^t, x^t\rangle
\end{equation}
At each iteration, our online \buyers algorithm calls projected gradient descent and receives a prediction $x^t$. Then our algorithm draws a set $S^t\sim D(x^t)$. Then it feeds back to projected gradient descent the parameter vector $\theta^t$ that it observed or equivalently the function $f^t$ that it observed. 

We first analyze how projected gradient works, so as to show that we can run the algorithm in polynomial time. Then we state the algorithms regret guarantees and show how these guarantees translate to guarantees about the online \buyers problem.

Projected gradient descent works as follows: Let $\Pi_{\Po}(y)=\arg\min_{x\in \Po}\|x-y\|_2$ denote the $L_2$ projection of $y$ onto $\Po$. At each iteration the projected gradient descent plays:
\begin{equation}
x^t=\Pi_{\Po}\left(x^{t-1}+\eta_t \nabla f^{t-1}(x^{t-1})\right)=\Pi_{\Po}\left( x^{t-1}+\eta_t \nabla F(x^{t-1}) - \eta_t \theta^t\right).
\end{equation}
Our polytope $\Po=[0,1]^m$ is so simple that the $L_2$ projection step takes the closed form:
\begin{equation}
\Pi_{\Po}^j(y) = \max\{0,\min\{1,y_j\}\}
\end{equation}
and thereby is efficiently computable. The gradient step is also efficiently computable if we can compute the gradient of $F$ at any point $x$ in polynomial time.

\cite{Zinkevich2003} provided the first bounds on projected gradient descent. We borrow the slightly improved bounds of \cite{Hazan2006}. If we denote with $D=\max_{x,y\in \Po}\|x-y\|_2$ and with $G=\sup_{x\in \Po, t\in [T]} \|\nabla f^t(x)\|_2$, then we get that that the average regret of projected gradient descent with $\eta_t=\frac{G}{D\sqrt{t}}$ is upper bounded by $3GD\sqrt{\frac{1}{T}}$, i.e. for any sequence of thresholds $\theta^{1:T}$:
\begin{equation}
\max_{x^*\in \Po}\frac{1}{T} \sum_{t=1}^T \left(f^t(x^*)-f^t(x^t)\right) \leq 3GD \sqrt{\frac{1}{T}}
\end{equation}
In our setting, we have that: $D = \max_{x,y\in [0,1]^m}\|x-y\|_2\leq \sqrt{m}$ and we have that:
\begin{align*}
\|\nabla f^t(x)\|_2 = \| \nabla F(x) - \theta^t\|_2\leq \|\nabla F(x)\|_2 + \|\theta^t\|_2
\end{align*}
Thus if we denote with $H=\sup_{x} \|\nabla F(x)\|_2$ and with $K=\max_{t\in [T], j\in [m]}\theta_j^t$, then we get that:
\begin{equation}
\max_{x^*\in \Po}\frac{1}{T} \sum_{t=1}^T \left(f^t(x^*)-f^t(x^t)\right) \leq 3(H+\sqrt{m K})\sqrt{\frac{m}{T}}
\end{equation}

Finally, we show how this regret guarantee of projected gradient descent maps back to a guarantee for our algorithm for the online \buyers problem. Observe that since by the definition of a convex rounding scheme, $y_j(x)\leq x_j$, we have that at each time-step $t$:
\begin{equation}
\E_{S^t\sim D(x^t)}\left[u(S^t,\theta^t)\right] = F(x^t) - \sum_{j\in [m]} \theta_j^t \cdot y_j(x^t) \geq F(x^t) - \sum_{j\in [m]} \theta_j^t \cdot x_j^t = F(x^t) - \langle \theta^t, x^t\rangle
\end{equation}
Combining the latter with the regret guarantees of projected gradient descent we get that for any sequence of adaptively chosen $\theta^{1:T}$: 
\begin{align*}
\frac{1}{T}\sum_{t=1}^T\E_{S^t\sim D(x^t)}\left[u(S^t,\theta^t)\right] \geq~& \frac{1}{T} \sum_{t=1}^T \left(F(x^t) - \langle \theta^t, x^t\rangle\right)\\
\geq~& \max_{x\in \Po} \left( F(x) -  \left\langle\frac{1}{T} \sum_{t=1}^T\theta^t, x\right\rangle \right)- 3(H+\sqrt{m K})\sqrt{\frac{m}{T}}
\end{align*}
Invoking the $\alpha$-approximate property of the convex rounding scheme we then conclude: 
\begin{align*}
\frac{1}{T}\sum_{t=1}^T\E_{S^t\sim D(x^t)}\left[u(S^t,\theta^t)\right]\geq~& \max_{x\in \Po} \left( F(x) -  \left\langle\frac{1}{T} \sum_{t=1}^T\theta^t, x\right\rangle \right)- 3(H+\sqrt{m K})\sqrt{\frac{m}{T}}\\
\geq~&\max_{S} \left( \frac{1}{\alpha}v(S) -  \sum_{j\in S} \hat{\theta}_j^T \right)- 3(H+\sqrt{m K})\sqrt{\frac{m}{T}}
\end{align*}
\end{proof}

\cite{Dughmi2011} gave a $\left(1-\frac{1}{e}\right)$-approximate convex-rounding scheme for the case of coverage valuations, through the process that they call Poisson rounding. 
\begin{defn}[Poisson rounding] Given an input $x$ define $D(x)$ as follows: for each item $j$ independently include $j$ in the output set with probability: $1-\exp\{-x_j\}$. 
\end{defn}

We show that it satisfies all the properties that we need for the case of explicitly given coverage valuations. Some of these properties were not shown in \cite{Dughmi2011} as they were not needed, hence we provide a complete proof. 
\begin{lemma}\label{lem:poisson-convex} Poisson rounding is a $\frac{e}{e-1}$-approximate convex rounding scheme for explicitly given coverage valuations. 
Moreover, $\sup_{x\in \Po} \|\nabla F(x)\|_2 \leq \max_{j\in [m]}v(\{j\})\sqrt{m}$.
\end{lemma}
\begin{proof}
The following properties where shown by \cite{Dughmi2011}, hence we re-direct the reader to that paper: i) the function $F(x) = \E_{S\sim D(x)}[v(S)]$ is concave, ii) it has a polynomially computable gradient given the hyper-graph representation of a coverage valuation as input, iii) the scheme satisfies that for any integral $x$ associated with a set $S$, $F(x)\geq \left(1-\frac{1}{e}\right)v(S)$. Last, it is easy to see that $y_j(x) = 1-\exp\{-x_j\}\leq x_j$. Thus all four properties of a convex rounding scheme are satisfied.

It remains to upper bound the quantity $\sup_{x\in \Po}\|\nabla F(x)\|_2$. For an explicitly given weighted hypergraph representation $\G=(V,E)$ of a coverage valuation we have:
\begin{equation}
F(x) = \sum_{v\in V} w_v \left(1-\prod_{j\in E: v\in j} \exp\{-x_j\}\right)=\sum_{v\in V} w_v \left(1-\exp\left\{-\sum_{j\in E: v\in j} x_j\right\}\right)
\end{equation}
Thus we have that:
\begin{equation}
\nabla_j F(x) = \sum_{v\in V: v\in j} w_v \exp\left\{ -\sum_{j'\in E: v\in j'} x_{j'}\right\} \leq \sum_{v\in V: v\in j} w_v = v(\{j\})
\end{equation}
Hence:
\begin{equation}
\|\nabla F(x)\|_2 \leq \sqrt{\sum_{j\in [m]} v(\{j\})^2}\leq \max_{j\in [m]} v(\{j\})\cdot \sqrt{m}
\end{equation}
\end{proof}

Combining Lemma~\ref{lem:online-convex} and Lemma~\ref{lem:poisson-convex} we get Lemma~\ref{lem:buyers-learning}.

\section{Oracle Learning and Finite Parameter Space}
\label{sec:app-finite}\label{sec:finite-parameter}

In this section we examine a special case of the general online learning problem, defined in Section \ref{sec:oracles}, with the restriction that the adversary can pick among only a small set of $d$ parameters. Then we show how this general results implies specific results for no-regret learning in the online bidding problem and no-regret learning in the security games model of \cite{Balcan2015}. Specifically, we show that in any online learning problem where the adversary has only $d$ choices to pick from on each day, i.e. $|\Theta|=d$, then with an access to an offline optimization problem an instantiation of Algorithm~\ref{defn:lazy-ftpl}, can achieve regret of $2H\sqrt{d/T}$, assuming utilities are bounded in $[0,H]$. This regret is independent of the number of available actions! Our result could be of independent interest and nicely complements recent impossibility results of \cite{Hazan2015}. 


\paragraph{Regret Analysis.} We will assume that $u(a,\theta)\in [0,H]$, which is satisfied in our online bidding problem for any non-overbidding action $a$ and when the value of the player is at most $H$. Observe that when the parameter space is finite, then we can map each sequence $\theta^{1:t}$ to a vector $\phi^t\in \N_+^d$, representing how many times each parameter $\theta\in \Theta$ has arrived in the first $t$ time-steps. Then we will denote with $\phi^t$ this frequency vector at time-step $t$. Thus we can equivalently write the cumulative utility from a fixed action in hindsight in terms of this frequency vector as:
\begin{equation}
U(a,\phi^t) = \sum_{\theta\in \Theta} \phi_{\theta}^t\cdot  u(a,\theta)
\end{equation}
We will interchangeably use this notation and the old $U(a,\theta^{1:t})$ notation in this section. Equivalently, we will be writing for any frequency vector $\phi\in \N_{+}^d$: 
\begin{equation}
M(\phi) = \arg\max_{a\in A} U(a,\phi)
\end{equation}
for the optimization oracle as a function of the frequency vector, rather than the actual sequence and we will assume that we have access to such an oracle $M$.

We will show that, given such an optimization oracle, there exists an efficient no-regret algorithm which achieves regret $O(H\sqrt{d T})$ against adaptive adversaries. We consider the instantiatiation of Algorithm~\ref{defn:lazy-ftpl} with the following sample perturbation.
\begin{defn}[Geometric sample perturbations] We construct the random sequence $\{x\}$ as follows: we go over each possible parameter $\theta\in \Theta$. For each possible parameter vector we flip a coin with probability $p$ of heads. If it comes up tails then we add an extra fake observation of that parameter to the sequence $\{x\}$ and continue flipping the coin adding an extra parameter at each time we see tails. If the coin ever comes up heads, then we stop adding vectors from that parameter and we move on to the next parameter.
\end{defn}
Typically we will be setting the parameter $p$ of the order of $1/\sqrt{T}$, so that in expectation we add $\sqrt{T}$ extra occurrences to each possible parameter. 
An equivalent way of viewing our sample perturbations is as if we are adding a random vector $z$ to the frequency vector $\phi^t$, where each coordinate $z_\theta$ of $z$ is drawn independently from a geometric distribution with parameter $p$ and supported on $\{0,1,2,\ldots\}$. Under this interpretation at each time-step $t$ the algorithm is picking action $a^t=M(\phi^{t-1}+z)$. We will interchangeably use this interpretation of the perturbations in the following proof.

\begin{theorem}[Efficient oracle-based no-regret learning]\label{thm:finite-parameter-regret}
Follow the perturbed leader with geometric sample perturbations is efficiently computable, given access to an optimization oracle, and achieves regret against any adaptive adversary:
\begin{equation}
\max_{a\in A} \sum_{t=1}^T \E\left[u(a,\theta^t)-u(a^t,\theta^t)\right]\leq 2H\sqrt{d T}
\end{equation}
\end{theorem}

\subsection{Application to No-Regret in Online Bidding}
In this section we examine whether we can actually achieve the stronger no-regret property, rather than the no-envy property, in the online bidding problem against adversaries, assuming that we are given an offline bidding optimization oracle. 

An easy corollary of our general theorem in the previous section is that if the adversary can pick one among $d$ possible threshold vectors on each day, then there exists an efficient no-regret algorithm for the online bidding problem against adversaries, assuming that we have oracle access to the offline optimization problem. Observe that this problem is highly non-trivial, as without the optimization oracle, our impossibility result covers such settings, since in the set cover reduction, the number of different price vectors used was at most $r\cdot m$, which is finite and polynomial in the description of the instance. Yet, without an oracle there does not exist any polynomial time no-regret algorithm. We will show that in such settings an optimization oracle bypasses this impossibility.

\begin{corollary}\label{thm:bidding-oracles}
Consider a bidder with any valuation $v(\cdot)\leq H$ participating in a sequence of simultaneous second-price auctions and assume that the number of different threshold vectors that can arise each day is at most $d$. Assuming access to a bidding oracle against explicitly given distributions, there exists a polynomial-time computable learning algorithm for computing the bidder's bid vector $b^t$ at every time step $t$ such that after $T$ iterations the regret of the bidder is bounded by:
\begin{equation}
\E_{b^{1:T}}\left[\frac{1}{T}\sum_{t=1}^{T} \left(u(b^*,\theta^t)-u(b^t,\theta^t)\right)\right] \leq 2H\sqrt{\frac{d}{T}},
\end{equation}
\end{corollary}

\subsection{Application to Security Games}
\label{sec:app-security}

Our general online learning results in Section~\ref{sec:finite-parameter} find a nice application in the context of security games and specifically in the repeated Stackelberg model of \cite{Balcan2015}. In this model a defender repeatedly plays every day against a sequence of attackers. On each day he commits to a mixed strategy $a\in \Delta(S)$ on protecting a set of $n$ targets, i.e. $S=2^{[n]}$. The attacker that arrives each day is characterized by a type $\theta$, which comes from a finite type space $\Theta$, with $|\Theta|=k$. Based on his type, the attacker best responds to the mixed strategy $a$ of the defender, through some arbitrary best-response function. Then the utility of the defender is some function of the best-response action of the attacker and the action that he committed. Observe that the action space of the defender is infinite, as it is the space of mixed strategies. \cite{Balcan2015} show that only a finite subset of the actions make sense and then invoke standard learning algorithms, resulting in a regret rate of $O\left(\sqrt{\frac{n^2k\log(nk)}{T}}\right)$, assuming utilities of the defender are bounded in $[0,1]$. Their algorithm however is not polynomially computable as it invokes standard learning algorithms over an exponential action space.

The crucial property is that the utility of the defender on each day can be viewed as only a function of the action $a$ that he committed to and the type of the attacker $\theta$, i.e. $u(a,\theta)$. This function is a complicated non-linear function, but our results in Section~\ref{sec:finite-parameter} do not require any structure from this function. Hence, by applying our Theorem~\ref{thm:finite-parameter-regret} we get the following theorem:
\begin{corollary}
The perturbed leader algorithm with geometric sample perturbations achieves regret rate of $2\sqrt{\frac{k}{T}}$, when applied to the Stackelberg security game of \cite{Balcan2015}. Moreover, it runs in polynomial time assuming an offline best-commitment oracle against an explicitly given distribution of attacker types. 
\end{corollary}
In other words, we get a much better regret bound and more importantly a bound that is independent of the number of targets $n$!

Moreover, to apply our algorithm we do not need a pre-processing step, as in \cite{Balcan2015}, to make the action space of the defender finite. The latter is handled within the offline oracle. One way to construct an offline oracle is to utilize the pre-processing step of \cite{Balcan2015} and then go through all actions in the reduced action space. This would be a finite yet exponential running time, which is also the running time of the algorithms in \cite{Balcan2015}. Recently, Li et al. \cite{Li2016}, has shown that in its general form this online problem is NP-hard, thereby the exponential computation is insurmountable in the general form of the problem. However, our result allows for the immediate use of any better polynomial time offline algorithms for specific settings, in a black-box manner.

\subsection{Omitted Proofs}
\begin{proofof}{Theorem~\ref{thm:finite-parameter-regret}}
First by applying Theorem~\ref{thm:general-regret}, we get that the regret of the algorithm is bounded by: 
\begin{equation}
\E_{\{x\}}\left[\sum_{t=1}^T \left( u(a^*,\theta^t) - u(a^t,\theta^t)\right)\right] \leq \sum_{t=1}^T g(t) +  H\cdot \E[\|z\|_1] =  \sum_{t=1}^T g(t) +  H\cdot \frac{d}{p}
\end{equation}
where we used the fact that $u(a,\theta)\in [0,H]$ and thereby, $f_-(\theta)=0$ and $f_+(\theta)=H$ and the fact that $\|z\|_1$ is the sum of $d$ independent geometrically distributed random variables each with parameter $p$, thereby $\E[\|z\|_1]=\frac{d(1-p)}{p}\leq \frac{d}{p}$.

Now it suffices to bound the stability $g(t)$ of the algorithm, which is the main technical difficulty and which we will argue in the following stability lemma:
\begin{lemma}[Stability lemma] For any sequence $\theta^{1:T}$ and for any time-step $t$:
\begin{equation}
\E_{z}\left[ u(M(\phi^{t-1}+z),\theta^t)\right] \geq (1-p)\cdot \E_{z}\left[u(M(\phi^{t}+z),\theta^t)\right]
\end{equation}
Hence, $\E_{z}\left[ u(M(\phi^{t}+z),\theta^t)- u(M(\phi^{t-1}+z),\theta^t)\right]\leq p\cdot H = g(t)$.
\end{lemma}
\begin{proof}
Let $y^t\in \N^d$ denote a vector which has $1$ on the coordinate associated with $\theta^t$ and $0$ on every other coordinate. We can write:
\begin{align*}
\E_{z}\left[ u(M(\phi^{t}+z),\theta^t)\right] =~& 
\sum_{k\in \N_{\geq 0}^d} \Pr[z=k]\cdot u(M(\phi^{t}+k),\theta^t)\\
 =~& \sum_{k\in \N_{\geq 0}^d} \Pr[z=k]
\cdot u(M(\phi^{t-1}+k+y^t),\theta^t)\\
=~&\sum_{\tilde{k}\in \N_{\geq 0}^d, \tilde{k}_{\theta_t}\geq 1} \Pr[z=\tilde{k}-y^t]
\cdot u(M(\phi^{t-1}+\tilde{k}),\theta^t)
\end{align*}
Observe that when each coordinate of $z$ is an independent geometric distribution then for any $\tilde{k}\in \N_{\geq 0}^d$ such that $\tilde{k}_{\theta_t}\geq 1$:
\begin{equation}
\Pr[z=\tilde{k}-y^t] = p^n (1-p)^{\|\tilde{k}-y^t\|_1} = (1-p)^{-1} p^n (1-p)^{\|k\|_1} = (1-p)^{-1}\Pr[z=k]
\end{equation}
Thus we get:
\begin{align*}
\E_{z}\left[ u(M(\phi^{t}+z),\theta^t)\right]=~& 
(1-p)^{-1}\sum_{\tilde{k}\in \N_{\geq 0}^d, \tilde{k}_{\theta_t}\geq 1} \Pr[z=\tilde{k}]
\cdot u(M(\phi^{t-1}+\tilde{k}),\theta^t)\\
\leq~&(1-p)^{-1}\sum_{\tilde{k}\in \N_{\geq 0}^d} \Pr[z=\tilde{k}]
\cdot u(M(\phi^{t-1}+\tilde{k}),\theta^t)\\
=~&(1-p)^{-1}\E_{z}\left[u(M(\phi^{t-1}+z),\theta^t)\right]
\end{align*}
\end{proof}
Thus we can conclude that the regret of the algorithm is bounded by:
\begin{equation}
\E_{\{x\}}\left[\sum_{t=1}^T \left( u(a^*,\theta^t) - u(a^t,\theta^t)\right)\right] \leq T \cdot H\cdot p + H\cdot \frac{d}{p}
\end{equation}
Setting $p=\sqrt{\frac{d}{T}}$, yields regret $2H\sqrt{d T}$, which completes the proof.
\end{proofof}

%
\end{appendix}

\end{document}